%% file: main.tex
\documentclass[lettersize,journal]{IEEEtran}

\usepackage{amsmath,amsfonts}
\usepackage{algorithmic}
\usepackage{algorithm}
\usepackage{array}
\usepackage{textcomp}
\usepackage{stfloats}
\usepackage{url}
\usepackage{verbatim}
\usepackage{graphicx}
\usepackage{cite}
\DeclareGraphicsExtensions{.pdf,.jpeg,.png,.es,.jpg,.eps}
\usepackage{caption}
\usepackage{subfigure}

\usepackage{amssymb}
\usepackage{amsthm}
\usepackage{dsfont}
\usepackage{cite}
\usepackage{mathtools}
\usepackage{threeparttable}
\usepackage{colortbl}
\usepackage{siunitx}
\DeclareSIUnit[]{\pu}{p.u.}
\DeclareSIUnit[]{\VA}{VA}

\DeclareSymbolFont{bbold}{U}{bbold}{m}{n}
\DeclareSymbolFontAlphabet{\mathbbold}{bbold}

\newcommand{\diag}[1]{\ensuremath{\mathrm{diag}(#1)}}
\newcommand{\range}[1]{\ensuremath{\mathrm{range}(#1)}}

\DeclarePairedDelimiterX\Set[2]{\lbrace}{\rbrace}%
{ #1 \,\delimsize| \,\mathopen{} #2 }

\DeclareMathOperator\sign{sign}

\newcommand{\real}[0]{\mathbb R}

\usepackage{tikz}
\newcommand*\circled[1]{\tikz[baseline=(char.base)]{\node[shape=circle,draw,inner sep=0.05pt] (char) {#1};}}
\newtheoremstyle{bfnote}%
{}{}%
{\itshape}{}%
{\bfseries}{.}%
{ }%
{\thmname{#1}\thmnumber{ #2}\thmnote{ (#3)}}
\theoremstyle{bfnote}
\newtheorem{thm}{Theorem}
\newtheorem{cor}{Corollary}
\newtheorem{rem}{Remark}
\newtheorem{lem}{Lemma}
\newtheorem{ass}{Assumption}

\usepackage{enumitem}
\setlist[itemize]{leftmargin=*}

\begin{document}

\title{Stable Reinforcement Learning for Optimal Frequency Control: A Distributed Averaging-Based Integral Approach}

\author{Yan Jiang, Wenqi Cui, Baosen Zhang, and Jorge Cort\'es
\thanks{Y. Jiang, W. Cui, and B. Zhang are with the Department
of Electrical and Computer Engineering, University of Washington, Seattle, WA 98195. Emails: {\tt {\{jiangyan,wenqicui,zhangbao\}@uw.edu}.} J. Cort\'es is with the Department of Mechanical and Aerospace Engineering, University of California, San Diego, CA 92093. Email: {\tt{cortes@ucsd.edu}.}}
\thanks{Y. Jiang, W. Cui, and B. Zhang are partially supported by NSF grants ECCS-1807142, ECCS-1942326 and ECCS-2153937. J. Cor\'tes is partially supported by NSF grant ECCS-1947050.}}



\maketitle

\begin{abstract}
Frequency control plays a pivotal role in reliable power system operations. It is conventionally performed in a hierarchical way that first rapidly stabilizes the frequency deviations and then slowly recovers the nominal frequency. However, as the generation mix shifts from synchronous generators to renewable resources, power systems experience larger and faster frequency fluctuations due to the loss of inertia, which adversely impacts the frequency stability. This has motivated active research in algorithms that jointly address frequency degradation and economic efficiency in a fast timescale, among which the distributed averaging-based integral (DAI) control is a notable one that sets controllable power injections directly proportional to the integrals of
frequency deviation and economic inefficiency signals. Nevertheless, DAI do not typically consider the transient
performance of the system following power disturbances and has been restricted to quadratic operational cost functions. This manuscript aims to leverage nonlinear optimal controllers to simultaneously achieve optimal transient frequency control and find the most economic power dispatch for frequency restoration. To this end, we integrate reinforcement learning (RL) to the classic DAI, which results in RL-DAI. Specifically, we use RL to learn a neural network-based control policy mapping from the integral variables of DAI to the controllable power injections which provides optimal transient frequency control, while DAI inherently ensures the frequency restoration and optimal economic dispatch.
Compared to existing methods, we provide provable guarantees on the stability of the learned controllers and extend allowable cost functions to a much larger class. Simulations on the $39$-bus New England system illustrate our results.
\end{abstract}

\begin{IEEEkeywords}
Frequency control, Lyapunov stability, reinforcement learning, steady-state and transient performance.
\end{IEEEkeywords}

\section{Introduction}
\input{intro.tex}

\section{Modelling Approach and Problem Statement}\label{Sec:modelling}
\input{model.tex}

\section{Generalized Distributed Averaging-Based Integral Control}\label{sec:Generalized-DAI}
\input{DAI.tex}

\section{Reinforcement Learning for Optimal Transient Frequency Control}\label{sec:RL}
\input{RL.tex}

\section{Numerical Illustrations}\label{sec:simu}
\input{simulation.tex}

\section{Conclusions}\label{sec:conclusion}
We have developed a framework that bridges the gap between the optimization of steady-state and transient frequency control performance. Our approach combines the study of the controller properties that make it stabilizing with reinforcement learning techniques to find the stable nonlinear controller that optimizes the transient frequency behavior. We have then integrated the learned controller into a distributed averaging-based integral mechanism that optimizes the steady-state cost of restoring frequency to their nominal values. Several interesting research directions are open. We would like to expand the characterization of the properties of the learned controller beyond its stabilizing nature to address questions about robustness to unmodeled dynamics. In addition, the controllers only require local information once they are implemented, but they are trained in a centralized manner. Understanding how they can be trained in a distributed way is also very relevant. Finally, in this paper, we have extended the class of cost functions from quadratic to functions that satisfy a certain scaling property. It would be interesting to see if distributed averaging integral control can be applied to an even larger class of cost functions. 


\appendix[]
\input{appendix.tex}

%
%

\bibliographystyle{IEEEtran}
\bibliography{main}


 





\end{document}

%% file: intro.tex
\IEEEPARstart{T}{he} key to the normal operation of a power system is the balance between electric power supply and demand over the network~\cite{Kirschen2019eco}. For instance, the main cause of the 2021 Texas power crisis is that the deficient supply of power due to frozen equipment could not meet the high demand for electricity in cold weather. A system frequency deviation from its nominal value is a reflection of a power imbalance~\cite{Eto2010lbnl}, which makes frequency control a vital task of grid operators. Traditionally, this task is performed in a hierarchical structure composed of three layers with timescale separation: primary---droop control ($\SI{<20}{\second}$), secondary---frequency restoration ($\SI{30}{\second}$--$\SI{10}{\minute}$), and tertiary---economic dispatch ($\SI{>15}{\minute}$)~\cite{Eto2010lbnl}. 

Nowadays, power systems are experiencing a change in the mix of generation, where conventional synchronous generators are gradually being replaced by renewable energy sources like solar and wind energy~\cite{benjamin2017}. It is anticipated that the renewable share of the electricity generation mix in the United States will double from $21\%$ in 2020 to $42\%$ in 2050~\cite{eia2021outlook}. Intermittent renewable sources are typically inverter-interfaced, which may adversely affect the robustness of the frequency dynamics due to the loss of inertia~\cite{nrel2013renew}. This exposes power systems to larger and faster frequency fluctuations than before, which has motivated
active research on flexible distributed frequency control schemes that can break the hierarchy by addressing simultaneously frequency degradation
%
%
and economic efficiency at a fast timescale.

A key challenge in frequency control is that the power imbalances across the network are not explicitly known. A number of studies have proposed distributed algorithms to overcome this challenge. The works in~\cite{Zhao2015acc, li2016tcns, mallada2017optimal, DORFLER2017Auto, SCHIFFER2017auto, weitenberg2018exponential, Wang2019tsg, you2021tac,TS-AC-CDP-AJS-JC:21-auto} focus on optimizing the steady-state frequency and economic performance using a principled design of fixed updating rules involving agent communication for control dynamics. The approaches mainly fall in two categories. The first category~\cite{li2016tcns, mallada2017optimal, Wang2019tsg, you2021tac,TS-AC-CDP-AJS-JC:21-auto} rests on a primal-dual interpretation of power system dynamics under a properly designed optimization problem. This approach, however, always requires the estimation of certain system parameters. The second category~\cite{Zhao2015acc, DORFLER2017Auto, SCHIFFER2017auto, weitenberg2018exponential} builds upon various consensus algorithms to converge to an equilibrium with nominal frequency and economic efficiency. A notable example in this category is the distributed averaging-based integral (DAI) mechanism~\cite{Zhao2015acc, Johannes2016ecc, SCHIFFER2017auto, weitenberg2018exponential}, where the controllable
power injections are directly proportional to the integrals of frequency deviation and economic inefficiency signals. A key caveat to this approach is that DAI control has been so far restricted to quadratic cost functions. 

The works above focus on the optimization of the steady-state performance and do not typically consider the transient performance along the system trajectories following power disturbances. In fact, the optimization of transient performance is a challenging problem due to the nonlinearity of power dynamics and the uncertainty in power disturbances. Reinforcement learning (RL)~\cite{Sutton2018RL} is a powerful tool for learning from interactions with uncertain environments and determining how to map situations to actions so that a desired performance is optimized. By virtue of the above feature, RL has emerged~\cite{Ernst2004tps,Chen_2022tsg}  as an effective instrument to address the optimal transient frequency control problem in nonlinear power systems under unknown power disturbances. Nevertheless, the Achilles' heel of standard RL algorithms is their lack of provable stability guarantees, which presents a significant barrier to its practical implementation for the operation of power systems. In fact, many works~\cite{Ye2011tps, Singh2017tie, Yan2020tps, Chen2021tii}
optimize the transient performance by learning control policies that exhibit good performance against data but without any provable guarantees on steady-state performance. The recent paper~\cite{cui2022tps} on RL for optimal primary frequency control proposes a way to address the stability issue by identifying a set of properties that make a control design stabilizing and then restricting the search space of neural network-based controllers. In this paper, we extend this idea to achieve provable guarantees on frequency restoration and economic efficiency at the steady-state. 
%
%

With the aim of filling the gap between the optimization of steady-state and transient frequency control performance, we propose to unify
these two perspectives by integrating RL into DAI control. In our approach, we employ RL to seek optimal control policies in terms of the transient performance as a map from the integral variables of DAI to the controllable power injections, while DAI inherently ensures the optimal steady-state performance. This results in a nonlinear optimal frequency controller which we term \emph{RL-DAI control} that generalizes the standard DAI approach.
%
%
Specifically, the generalization is manifold:
\begin{itemize}
\item Unlike the classic DAI control that only addresses quadratic operational cost functions,  RL-DAI control admits any strictly convex cost functions
whose gradients are identical up to heterogeneous scaling
factors;
\item Unlike RL methods that do not  provide stability guarantees, RL-DAI control encodes the stabilization requirement on the control policy as mild conditions on its continuity and monotonicity properties. Such conditions can be easily realized by an ingenious parameterization of the control policy as a monotonic neural network,
which is trained by a RL algorithm based on recurrent neural networks (RNNs);
%
%
\item RL-DAI control jointly optimizes both steady-state and transient frequency control performance by leveraging the added degrees of freedom in tuning parameters that characterize the nonlinear control policy.
\end{itemize}

The rest of this manuscript is organized as follows. Section~\ref{Sec:modelling}
describes the power system model and formalizes the optimal frequency control problem. Section~\ref{sec:Generalized-DAI} describes the proposed generalized DAI control and shows how it guarantees the asymptotic convergence of the closed-loop system to the equilibrium that achieves the steady-state
performance objectives. Section~\ref{sec:RL} illustrates how to integrate RL with DAI control such that the transient performance can be optimized without jeopardizing stability. Section~\ref{sec:simu} validates our results
through detailed simulations. We gather our conclusions and ideas for future work in Section~\ref{sec:conclusion}.

%% file: model.tex
In this section, we describe the power system model used in this paper for analysis and the frequency control problem we aim to address.

\subsection{Power System Model}
We consider\footnote{Throughout this manuscript, vectors are denoted in lower case bold and matrices are denoted in upper case bold, while scalars are unbolded, unless otherwise specified. Also, $\mathbbold{1}_n, \mathbbold{0}_n \in \real^n$ denote the vectors of all ones and all zeros, respectively.} a $n$-bus power system whose topology can be characterized by a weighted undirected connected graph $\left(\mathcal{V},\mathcal{E} \right)$, where buses indexed by $i,j \in \mathcal{N} :=\{1,\dots, n\} $ are linked through transmission lines denoted by unordered pairs $\{i,j\} \in \mathcal{E} \subset \Set*{\{i,j\}}{i,j\in\mathcal{N},i\not=j}$.

Within the set of buses $\mathcal{N}$, we consider two subsets: the buses with generators\footnote{The generator buses can have collocated loads}  $\mathcal{G}$ and buses that are purely loads $\mathcal{L}$. Generator buses are described with differential equations, while load buses are treated as frequency-sensitive loads governed by algebraic equations. 
%
%
More precisely, given net power injections $\boldsymbol{p} := \left(p_{i}, i \in \mathcal{N} \right) \in \real^n$
%
%
on each bus, the dynamics of the voltage angle $\theta_i$ (in $\SI{}{\radian/s}$) and the frequency deviation $\omega_i$ (in $\SI{}{\pu}$) from the nominal frequency $f_0$ ($\SI{50}{\hertz}$ or $\SI{60}{\hertz}$ depending on the particular system) evolve according to
\begin{subequations}\label{eq:sys-dyn-theta}
\begin{align}
    \dot{\theta}_i \!=&\!\ 2 \pi f_0 \omega_i\,,\!\!\!\!\!\!&&\!\forall i\!\in\!\mathcal{N}\!\!\,,\\
    m_i \dot{\omega}_i \!=\!& -\!
    \alpha_i \omega_i\!-\!\! \!\sum_{j=1}^n\! v_iv_jB_{ij}\!\sin{\left(\theta_i\!-\!\theta_j\right)}\! +\! p_{i} \!+\! u_{i}\,,\!\!\!\!\!\!&&\forall i\!\in\!\mathcal{G}\!\,,\\
    0 \!=&\! -\!
    \alpha_i \omega_i\!-\!\!\! \sum_{j=1}^n\! v_iv_jB_{ij}\!\sin{\left(\theta_i\!-\!\theta_j\right)} \!+\! p_{i} \!+\! u_{i}\,,\!\!\!\!\!\!&&\forall i\!\in\!\mathcal{L}\!\,,\end{align}
\end{subequations}
where the parameters are defined as: $m_i:=2H_i>0$ --- generator inertia constant (in $\SI{}{\second}$), $\alpha_i>0$ --- frequency sensitivity coefficient from generator or load (in $\SI{}{\pu}$) depending on whether $i \in \mathcal{G}$ or $i \in \mathcal{L}$, $v_i>0$ --- voltage magnitude (in $\SI{}{\pu}$), and $B_{ij}$ --- susceptance (in $\SI{}{\pu}$). The susceptances define coupling between buses in the system and satisfies $B_{ij}=B_{ji}>0$ if $\{i,j\} \in\mathcal{E}$ and $B_{ij}=B_{ji}=0$ if $\{i,j\} \not\in\mathcal{E}$. Finally, $u_{i}$ is a controllable power injection to be designed for secondary frequency control.

\begin{rem}[Model assumptions]\label{rem:model}
We make the following assumptions: 
\begin{itemize}
\item Lossless lines: $\forall \{i,j\}\in\mathcal{E}$, the line resistance is zero.
\item Constant voltage profile: $\forall i \in \mathcal{N}$, the bus voltage magnitude $v_i$ is constant. 
\item Decoupling: Reactive power flows do not affect bus voltage angles.
\item $\forall \{i,j\} \in\mathcal{E}$, the equilibrium bus voltage angle difference is within $\pm \pi/2$, i.e., $\forall \{i,j\} \in\mathcal{E}$, $|\theta_i^\ast-\theta_j^\ast|\in [0,\pi/2)$.
\end{itemize}
These assumptions are standard assumption made in the literature and are used in most works that study for frequency control in transmission networks (for more information, see~\cite{Zhao:2014bp, mallada2017optimal}).
\end{rem}

\subsection{Frequency Control}\label{ssec:goal}
The basic goal of frequency control is to keep the frequency of the power system close to its nominal value. We now illustrate this in more detail by following the same line of argument as in~\cite{Zhao2015acc, DORFLER2017Auto, weitenberg2018exponential, weitenberg2018robust}. 

Since the frequency dynamics depend only on the phase angle differences, for easy of analysis, we express the dynamics in center-of-inertia coordinates~\cite{weitenberg2018exponential, weitenberg2018robust} by making the following change of variables:
$$\delta_i:=\theta_i-\dfrac{1}{n}\sum_{j=1}^{n} \theta_j\,,\qquad\forall i\in\mathcal{N}\,.$$
Then, the dynamics in \eqref{eq:sys-dyn-theta} can be rewritten and stacked into a vector form as:
\begin{subequations}\label{eq:sys-dyn-vec}
\begin{align}
    \dot{\boldsymbol{\delta}} =&\ 2 \pi f_0\left(I_n-\frac{1}{n}\mathbbold{1}_n\mathbbold{1}_n^T\right) \boldsymbol{\omega}\,,\\
    \boldsymbol{M} \dot{\boldsymbol{\omega}}_\mathcal{G} =& -
    \boldsymbol{A}_\mathcal{G} \boldsymbol{\omega}_\mathcal{G}- \nabla_\mathcal{G} U (\boldsymbol{\delta}) + \boldsymbol{p}_\mathcal{G} + \boldsymbol{u}_\mathcal{G}\,,\\
    \mathbbold{0}_{|\mathcal{L}|} =& -
    \boldsymbol{A}_\mathcal{L} \boldsymbol{\omega}_\mathcal{L}- \nabla_\mathcal{L} U (\boldsymbol{\delta}) + \boldsymbol{p}_\mathcal{L} + \boldsymbol{u}_\mathcal{L}\,,\end{align}
\end{subequations}
with $$U(\boldsymbol{\delta}):=- \dfrac{1}{2}\sum_{i=1}^n \sum_{j=1}^n v_i v_j B_{ij}\cos{\left(\delta_i-\delta_j\right)}\,,$$
where $\boldsymbol{\delta}:=\left(\delta_i, i \in \mathcal{N} \right) \in \real^n$, $\boldsymbol{\omega}:=\left(\omega_i, i \in \mathcal{N} \right) \in \real^n$, $\boldsymbol{p}:=\left(p_i, i \in \mathcal{N} \right) \in \real^n$, $\boldsymbol{u} := \left(u_{i},i \in \mathcal{N} \right) \in \real^n$, $\boldsymbol{M} := \diag{m_i, i \in \mathcal{N}} \in \real^{n \times n}$, $\boldsymbol{A}_\mathcal{G} := \diag{\alpha_{i}, i \in \mathcal{G}} \in \real^{|\mathcal{G}| \times |\mathcal{G}|}$, $\boldsymbol{A}_\mathcal{L} := \diag{\alpha_{i}, i \in \mathcal{L}} \in \real^{|\mathcal{L}| \times |\mathcal{L}|}$. Here, a vector with a set as subscript denotes the subvector composed only of the elements from that set, e.g., $\boldsymbol{\omega}_\mathcal{G}:=\left(\omega_i, i \in \mathcal{G} \right) \in \real^{|\mathcal{G}|}$.\footnote{For a finite set $\mathcal{G}$, $|\mathcal{G}|$ denotes the cardinality of $\mathcal{G}$.}

The function $U(\boldsymbol{\delta})$ allows us to compactly summarize the power flow in the system, since the vector of power flows between the buses are given by $\nabla U$. The equilibria of~\eqref{eq:sys-dyn-vec} satisfy
\begin{subequations}\label{eq:sys-dyn-vec-ss}
\begin{align}
    \boldsymbol{\omega}^\ast =&\ \frac{1}{n}\mathbbold{1}_n\mathbbold{1}_n^T \boldsymbol{\omega}^\ast\,,\label{eq:omega-i-eq}\\
    \mathbbold{0}_{|\mathcal{G}|} =& -
    \boldsymbol{A}_\mathcal{G} \boldsymbol{\omega}_\mathcal{G}^\ast- \nabla_\mathcal{G} U (\boldsymbol{\delta}^\ast) + \boldsymbol{p}_\mathcal{G} + \boldsymbol{u}_\mathcal{G}^\ast\,,\label{eq:delta-i-eq-g}\\
    \mathbbold{0}_{|\mathcal{L}|} =& -
    \boldsymbol{A}_\mathcal{L} \boldsymbol{\omega}_\mathcal{L}^\ast- \nabla_\mathcal{L} U (\boldsymbol{\delta}^\ast) + \boldsymbol{p}_\mathcal{L} + \boldsymbol{u}_\mathcal{L}^\ast\,.\label{eq:delta-i-eq-l}\end{align}
\end{subequations}
Clearly, \eqref{eq:omega-i-eq} implies that there exists a scalar $\omega^\ast$ such that $\boldsymbol{\omega}^\ast=\mathbbold{1}_n\omega^\ast$. Then, \eqref{eq:delta-i-eq-g} and \eqref{eq:delta-i-eq-l} can be combined into 
\begin{align}\label{eq:equli-sync-omega}
 \mathbbold{0}_n =& -
    \boldsymbol{A} \mathbbold{1}_n\omega^\ast- \nabla U (\boldsymbol{\delta}^\ast) + \boldsymbol{p} + \boldsymbol{u}^\ast
    \end{align}
with $\boldsymbol{A} := \diag{\alpha_{i}, i \in \mathcal{N}} \in \real^{n \times n}$. Moreover, after premultiplying \eqref{eq:equli-sync-omega} by $\mathbbold{1}_n^T$, we can characterize $\omega^\ast$ as
\begin{align}\label{eq:omega-syn}
\omega^\ast = \dfrac{\sum_{i=1}^n p_i+\sum_{i=1}^n u_i^\ast}{\sum_{i=1}^n \alpha_i}\,,
\end{align}
where we have used the zero net power flow balance  $\mathbbold{1}_n^T\nabla U (\boldsymbol{\delta})=0$ of a lossless power system.

Observe from \eqref{eq:omega-syn} that, in the absence of frequency control, the system undergoing power disturbances $\boldsymbol{p}$ will synchronize to a nonzero frequency deviation, i.e., $\omega^{\ast} \neq 0$, since $\omega^{\ast} = 0$ only if $\sum_{i=1}^n p_i+\sum_{i=1}^n u_i^\ast=0$. Therefore, the main goal of frequency control is to regulate the frequency such that $\omega^{\ast} = 0$ by providing appropriate controllable power injections $\boldsymbol{u}$ to meet power disturbances $\boldsymbol{p}$.
Note that the existence of equilibrium points $(\boldsymbol{\delta}^\ast, \mathbbold{1}_n\omega^\ast)$ such that \eqref{eq:equli-sync-omega} holds for $\omega^{\ast} = 0$ is equivalent to the feasibility of the power flow equation
\begin{align}\label{eq:equli-sync-omega0}
 \nabla U (\boldsymbol{\delta}^\ast) =&  \ \boldsymbol{p} + \boldsymbol{u}^\ast\,,
    \end{align}
that is, there exist some $\boldsymbol{\delta}^*$ solving \eqref{eq:equli-sync-omega0}.  This is a standing assumption of this manuscript. In fact, if  \eqref{eq:equli-sync-omega0} is feasible, then there are many solutions it.  Consequently, we seek to select among the feasible solutions by optimizing certain performance metrics, which we describe next.

\subsection{Performance Assessment}\label{ssec:metric}
For the design of frequency control strategies, not only frequency performance but also economic factors must be taken into account. 
Moreover, the secure and efficient operation of power systems relies on properly controlled frequency and cost in both slow and fast timescales. Thus, we now introduce the frequency and economic performance metrics used in this manuscript based on different timescales. 

\subsubsection{Steady-State Performance Metrics} Our control objective in the long run is to achieve the nominal frequency restoration, i.e., $\boldsymbol{\omega}^\ast=\mathbbold{0}_n$, as well as the lowest steady-state aggregate operational cost $C(\boldsymbol{u}^\ast):=\sum_{i=1}^n C_i(u_i^{\ast})$, where the cost function $C_i(u_i)$ quantifies either the generation cost on a generator bus or the user disutility on a load bus for contributing $u_i$. This results in the following constrained optimization problem called \emph{optimal steady-state economic dispatch problem}:
\begin{subequations}\label{eq:opt-ss}
	\begin{eqnarray}
		\label{eq:edp2.a}
		\min_{\boldsymbol{u}^\ast} &&      C(\boldsymbol{u}^\ast):=\sum_{i=1}^n C_i(u_i^{\ast})    \\
		\label{eq:edp2.b}
		\mathrm{s.t.} &&  
		\sum_{i=1}^n p_i+\sum_{i=1}^n u_i^\ast=0 \,,
	\end{eqnarray}	\label{eq:edp2}%
\end{subequations}
where \eqref{eq:edp2.b} is a necessary constraint on the steady-state controllable power injections $\boldsymbol{u}^\ast$ in order to achieve $\boldsymbol{\omega}^\ast=\mathbbold{0}_n$ as discussed in Section~\ref{ssec:goal}. We adopt the standard assumption that  the cost function $C_i(u_i)$ is strictly convex and continuously differentiable~\cite{DORFLER2017Auto, you2021tac}. Then, the optimization problem \eqref{eq:opt-ss} has a convex objective function and an affine equality constraint. Thus, by the Karush-Kuhn-Tucker conditions~\cite[Chapter 5.5.3]{Boyd2004convex}, $\boldsymbol{u}^\ast$ is the unique minimizer of \eqref{eq:opt-ss} if and only if it ensures \emph{identical marginal costs}~\cite{Zhao2015acc, Johannes2016ecc, DORFLER2017Auto, SCHIFFER2017auto, weitenberg2018robust}, i.e., 
\begin{align}\label{eq:id-marginal}
\nabla C_i (u_i^\ast) = \nabla C_j (u_j^\ast)\,,\qquad\forall i,j\in\mathcal{N}\,.
\end{align}



\subsubsection{Transient Performance Metrics} Following sudden major power disturbances, the transient frequency dip can be large in the first few seconds, especially in low-inertia power systems. This may trigger undesired protection measures and even cause cascading failures. Thus, besides the steady-state performance, one should also pay attention to the transient frequency performance with moderate economic cost. With this aim, we define the following transient performance metrics evaluated along the trajectories of the system \eqref{eq:sys-dyn-theta}:
\begin{itemize}
\item \emph{Frequency Nadir} is the maximum frequency deviation from the nominal frequency on each bus during the transient response, i.e., 
\begin{equation}\label{eq:Nadir}
	\|\omega_i\|_\infty := \max_{t\geq0} |\omega_i(t)|\,.
\end{equation}
\item \emph{Economic cost} measures the average cost during time horizon $T$, i.e.
\begin{equation*}
\bar{C}_{i,T} :=  \dfrac{1}{T}\int_{0}^T C_i(u_i(t)) \mathrm{d}t\,. \end{equation*}
\end{itemize}
Then the \emph{optimal transient frequency control problem} becomes:
\begin{subequations}\label{eq:opt-ts}
	\begin{eqnarray}
		\label{eq:opt-ts-obj}
		\min_{\boldsymbol{u}} &&      \sum_{i \in \mathcal{G}} \|\omega_i\|_\infty +\rho\sum_{i=1}^n\bar{C}_{i,T}   \\
		\label{eq:opt-ts-con}
		\mathrm{s.t.} &&  
		\boldsymbol{u}\ \mathrm{stabilize}\ \eqref{eq:sys-dyn-theta}\,,
	\end{eqnarray}
\end{subequations}
%
%
where $\rho>0$ is the coefficient for tradeoff between the frequency performance and the economic cost. Note that we impose the stability requirement on $\boldsymbol{u}$ as a hard constraint in the optimization problem \eqref{eq:opt-ts}, which will play a pivotal role in its design.

\subsection{Reinforcement Learning for Frequency Control}
Our goal is to design an optimal stabilizing controller that brings the system \eqref{eq:sys-dyn-theta} to an equilibrium that restores the nominal frequency, i.e., $\boldsymbol{\omega}^\ast=\mathbbold{0}_n$, and solves the optimal steady-state economic dispatch problem  \eqref{eq:opt-ss}, while solving the optimal transient frequency control problem \eqref{eq:opt-ts} at the same time. A good starting point to achieve our steady-state control goal is the well-known DAI control. However, it is not straightforward how to also optimize the transient performance. This is hard to be done purely by conventional optimization methods since power systems are nonlinear and power disturbances are unknown. Therefore, we would like to integrate RL into DAI to jointly optimize steady-state and transient performance. 
%
%

%% file: DAI.tex
We start this section by briefly reviewing distributed averaging-based integral (DAI) control and then generalizing it to account for nonlinear control laws. We show that the closed-loop system is stable as long as the nonlinearity satisfies conditions on continuity and monotonicity. Interestingly, the steady-state performance objectives are still achieved under our proposed generalized DAI even for nonquadratic cost functions provided that the gradients of the cost functions on individual buses satisfy a certain condition, which we describe below.

\subsection{Review and Generalization of DAI}
DAI control~\cite{Zhao2015acc, Johannes2016ecc, SCHIFFER2017auto, weitenberg2018exponential} is an established choice of frequency control strategy  to meet the steady-state performance objectives discussed in Section~\ref{ssec:metric}.  However, most works restrict the cost functions to be quadratic, i.e., 
\begin{align}\label{eq:cost-quad}
    C_i(u_i) =\dfrac{1}{2}c_i u_i^2\,,\qquad\forall i\in\mathcal{N}\,,
\end{align}
where $c_i>0$ is the cost coefficient. With the cost functions given by \eqref{eq:cost-quad}, the power injections $\boldsymbol{u}$ in the system \eqref{eq:sys-dyn-vec} under DAI are
\begin{subequations}\label{eq:DAI}
\begin{align}\label{eq:u(s)-linear}
    u_i(s_i) = k_i s_i\,,\qquad\forall i\in\mathcal{N}\,,
\end{align}
with 
\begin{align}\label{eq:DAI-s}
    \dot {s}_{i} =&-2 \pi f_0\omega_i -c_i \sum_{j=1}^n Q_{ij}\left(c_i u_i-c_j u_j\right)\,,
\end{align}
\end{subequations}
where $k_i>0$ is a tunable control gain and $Q_{ij}\geq0$ is the weight associated with an undirected connected communication graph $\left(\mathcal{V},\mathcal{E}_\mathrm{Q} \right)$ (not necessarily the same as the physical system) such that $Q_{ij}=Q_{ji}>0$ if buses $i$ and $j$ communicate, i.e., $\{i,j\} \in\mathcal{E}_\mathrm{Q} \subset \Set*{\{i,j\}}{i,j\in\mathcal{N},i\not=j}$, and $Q_{ij}=Q_{ji}=0$ otherwise, $\forall i,j\in\mathcal{N}$. 

The core idea behind DAI is that the control dynamics can settle down only if the nominal frequency restoration, i.e., $\omega_i=0$, $\forall i\in\mathcal{N}$, and the identical marginal costs, i.e., $\nabla C_i (u_i) =c_i u_i=c_j u_j= \nabla C_j (u_j)$, $\forall i,j\in\mathcal{N}$, are simultaneously achieved. Thus ensuring  $\boldsymbol{\omega}^\ast=\mathbbold{0}_n$ and $\boldsymbol{u}^\ast$ is the solution of the optimal steady-state economic dispatch problem \eqref{eq:opt-ss}.

However, DAI in \eqref{eq:DAI} does not address the transient behavior of the system. Therefore, trajectories that eventually reaches the equilibrium may have poor transient performances. Thus, it is desirable to modify it so that we have better flexibility in improving the transient performance. Inspired by the fact that the control policy $u_i(s_i)$ of DAI in \eqref{eq:u(s)-linear} is directly proportional to the integral variable $s_i$, we generalize DAI by allowing $u_i(s_i)$ to be any Lipschitz continuous and strictly increasing function of $s_i$, with $u_i(0)=0$. This extends $u_i(s_i)$ from a linear function to any nonlinear function satisfying these requirements. The added nonlinearity provides us with a good way to optimize the transient frequency performance while remaining optimal for the steady-state economic dispatch. 

\begin{ass}[Continuity and monotonicity of $\boldsymbol{u}$]\label{ass:nonlinear-u}
$\forall i\in\mathcal{N}$, the function $u_i(\cdot):\real\mapsto \real$ is Lipschitz continuous and strictly increasing with $u_i(0)=0$.
\end{ass}
In Section~\ref{sec:RL}, we will show how this assumption can be imposed on controllers parameterized by neural networks. Namely, we can explicitly design the weights and biases of a neural network such that the required structure on the controller is explicitly satisfied. 

Next, we also relax the restriction to quadratic cost functions by allowing DAI to handle cost functions that satisfy the following assumption.

\begin{ass}[Scaled cost gradient functions]\label{ass-cost-grad-scale} There exists a strictly convex and continuously differentiable function $C_\mathrm{o}(\cdot):\real\mapsto \real$ and positive scaling factors $\boldsymbol{\zeta}:=\left(\zeta_i, i \in \mathcal{N} \right) \in \real^n_{+}$ such that
\begin{align} \label{eqn:cost_condition}
\nabla C_i(u_i)=\nabla C_\mathrm{o}(\zeta_i u_i)\,,\qquad \forall i\in\mathcal{N}\,.\end{align}
\end{ass}

Assumption~\ref{ass-cost-grad-scale} admits any strictly convex cost functions whose gradients are identical up to heterogeneous scaling factors. To look in greater detail at how this generalizes beyond quadratic cost functions, we provide some examples of common strictly convex functions that satisfy Assumption~\ref{ass-cost-grad-scale}. If the cost functions on individual buses are:
\begin{itemize}
\item \emph{power functions with positive even integer powers:}
    \begin{align}\label{eq:cost-power}
        C_i(u_i)= \dfrac{c_i}{r} u_i^r+b_i\,,
    \end{align}
where $c_i>0$ and $r$ is a positive even integer, then
    \begin{align}\label{eq:cost-power-grad}
\nabla C_i(u_i)=c_i u_i^{r-1}\,,\qquad\forall i\in\mathcal{N}\,. \end{align}
Thus, we can choose
\begin{align} \label{eq:zeta}
C_\mathrm{o}(\cdot)= \dfrac{1}{r}(\cdot)^r\qquad\ \text{and}\qquad\zeta_i=c_i^{\frac{1}{r-1}}\,,   
\end{align}
which satisfies the condition in~\eqref{eqn:cost_condition}.
\item \emph{arbitrary strictly convex and continuously differentiable functions that are identical up to constant terms:}
    \begin{align*}
        C_i(u_i)= C_\mathrm{o}(u_i)+b_i\,,
    \end{align*}
and its straightforward to check \eqref{eqn:cost_condition} is satisfied. 
\end{itemize}
Clearly, the quadratic cost functions in \eqref{eq:cost-quad} belong to the first scenario since \eqref{eq:cost-power} reduces exactly to \eqref{eq:cost-quad} if $r=2$ and $b_i=0$, which validates our generalization through Assumption~\ref{ass-cost-grad-scale}. 

Now, under Assumptions~\ref{ass:nonlinear-u} and~\ref{ass-cost-grad-scale}, we are ready to propose the \emph{generalized DAI} as follows: 
\begin{align}\label{eq:DAI-s-general}
    \dot {s}_{i} \!=\!-2 \pi f_0\omega_i \!-\!\zeta_i\!\! \sum_{j=1}^n\! Q_{ij}\!\left(\nabla C_i(u_i(s_i))\!-\!\nabla C_j(u_j(s_i))\right)\!\!\,.
\end{align}
Although not obvious at first sight, the generalized DAI in \eqref{eq:DAI-s-general} preserves the ability of the classic DAI to restore the nominal frequency and solves the optimal steady-state economic dispatch problem \eqref{eq:opt-ss} for cost functions satisfying Assumption~\ref{ass-cost-grad-scale}. In addition, the more general nonlinear form of $\boldsymbol{u}$ introduced in Assumption~\ref{ass:nonlinear-u} provides us with more degrees of freedom to better deal with the optimal transient frequency control problem \eqref{eq:opt-ts}. All of these statements would become clear as the analysis unfolds.

After combining \eqref{eq:sys-dyn-vec} and \eqref{eq:DAI-s-general}, we can write the overall closed-loop system dynamics under our proposed generalized DAI compactly as  
\begin{subequations}\label{eq:sys-dyn-vec-distribute}
\begin{align}
    \dot{\boldsymbol{\delta}} =&\ 2 \pi f_0\left(I_n-\frac{1}{n}\mathbbold{1}_n\mathbbold{1}_n^T\right) \boldsymbol{\omega}\,,\label{eq:sys-dyn-vec-distribute-delta}\\
    \boldsymbol{M} \dot{\boldsymbol{\omega}}_\mathcal{G} =& -
    \boldsymbol{A}_\mathcal{G} \boldsymbol{\omega}_\mathcal{G}- \nabla_\mathcal{G} U (\boldsymbol{\delta}) + \boldsymbol{p}_\mathcal{G} + \boldsymbol{u}_\mathcal{G}(\boldsymbol{s})\,,\\
    \mathbbold{0}_{|\mathcal{L}|} =& -
    \boldsymbol{A}_\mathcal{L} \boldsymbol{\omega}_\mathcal{L}- \nabla_\mathcal{L} U (\boldsymbol{\delta}) + \boldsymbol{p}_\mathcal{L} + \boldsymbol{u}_\mathcal{L}(\boldsymbol{s})\,,\label{eq:sys-dyn-vec-distribute-omegaL}\\
     \dot {\boldsymbol{s}} =& -2 \pi f_0 \boldsymbol{\omega} - \boldsymbol{Z}\boldsymbol{L}_\mathrm{Q}\nabla C(\boldsymbol{u}(\boldsymbol{s}))\,,\label{eq:bus-dyn}
     \end{align}
\end{subequations}
where $\boldsymbol{s}:=\left(s_i, i \in \mathcal{N} \right) \in \real^n$, $\boldsymbol{Z} := \diag{\zeta_i, i \in \mathcal{N}} \in \real^{n \times n}$, $\boldsymbol{L}_\mathrm{Q}:=\left[ L_{\mathrm{Q},ij}\right]\in \real^{n \times n}$ is the Laplacian matrix associated with the communication graph whose $ij$th element is
\begin{align*}
    \forall i,j \in \mathcal{N},\qquad L_{\mathrm{Q},ij}:=
    \begin{dcases}
    -Q_{ij} & \text{if}\ i\neq j\\
    \sum_{j^\prime=1,j^\prime\neq i}^n Q_{ij^\prime}& \text{if}\ i=j
    \end{dcases}\ \;,
\end{align*} 
and $\boldsymbol{u}(\boldsymbol{s}) := \left(u_{i}(s_i),i \in \mathcal{N} \right) \in \real^n$ with each $u_i(s_i)$ satisfying Assumption~\ref{ass:nonlinear-u}.

We conclude this subsection by stating a useful property of the inverse functions of the cost gradient functions on individual buses, which is a consequence of Assumption~\ref{ass-cost-grad-scale}. We make use of it later in 
the stability analysis of the closed-loop system~\eqref{eq:sys-dyn-vec-distribute}.

\begin{lem}[Identical scaled controllable power injections]\label{lem:inverse-id} If Assumption~\ref{ass-cost-grad-scale} holds, then  
\begin{align}\label{eq:inverse-id}
\zeta_i \nabla C_i^{-1}(\cdot)=\nabla C_\mathrm{o}^{-1}(\cdot)\,,\qquad \forall i\in\mathcal{N}\,,\end{align}
where $\nabla C_i^{-1}(\cdot)$ and $\nabla C_\mathrm{o}^{-1}(\cdot)$ denote the inverse functions of $\nabla C_i(\cdot)$ and $\nabla C_\mathrm{o}(\cdot)$, respectively.
\end{lem}
\begin{proof}
See the Appendix \ref{app:lem1-pf}.
\end{proof}

Provided that the cost functions on individual buses satisfy Assumption~\ref{ass-cost-grad-scale}, Lemma~\ref{lem:inverse-id} indicates that, for any given marginal price, each generator or load contributes the same amount of controllable power injection up to a scaling factor. 

\begin{rem}[Existence and monotonicity of $\nabla C_i^{-1}(\cdot)$ and $\nabla C_\mathrm{o}^{-1}(\cdot)$]\label{rem:inverse} Note that the existence of $\nabla C_i^{-1}(\cdot)$ is guaranteed by the strict convexity of $C_i(\cdot)$, which implies that $\nabla C_i(\cdot)$ is a strictly increasing function~\cite{rockafellar1998}. By a standard result in analysis, any strictly increasing function has an inverse that is also strictly increasing~\cite{binmore1977mathanalysis}. 
\end{rem}

\subsection{Closed-Loop Equilibrium Analysis}
The first result we show is that the steady-state performance objectives are still achieved under the generalized DAI. This is captured by the following theorem. 

\begin{thm}[Closed-loop equilibrium]\label{thm:equilibrium}
Suppose Assumptions~\ref{ass:nonlinear-u} and~\ref{ass-cost-grad-scale} hold. Assume $\forall \{i,j\} \in\mathcal{E}$, $|\delta_i^\ast-\delta_j^\ast|\in [0,\pi/2)$.\footnote{This is the counterpart of the model assumption, $\forall \{i,j\} \in\mathcal{E}$, $|\theta_i^\ast-\theta_j^\ast|\in [0,\pi/2)$ in Remark~\ref{rem:model}.} Then the equilibrium $(\boldsymbol{\delta}^\ast, \boldsymbol{\omega}^\ast, \boldsymbol{s}^\ast)$ of the closed-loop system \eqref{eq:sys-dyn-vec-distribute} is a unique point satisfying
\begin{subequations}\label{eq:equ-DAI} 
\begin{align}
\boldsymbol{\omega}^\ast=&\ \mathbbold{0}_n\,,\\ \nabla U (\boldsymbol{\delta}^\ast) =&\ \boldsymbol{p} + \boldsymbol{u}(\boldsymbol{s}^\ast)\,,\label{eq:equ-DAI-delta}\\ \boldsymbol{u}(\boldsymbol{s}^\ast) =&\ \nabla C_\mathrm{o}^{-1}(\gamma)\boldsymbol{Z}^{-1}\mathbbold{1}_n\,, \label{eq:equ-DAI-u}   \end{align}
\end{subequations}
with $\gamma$ uniquely determined by
\begin{align}\label{eq:gamma}
\nabla C_\mathrm{o}^{-1}(\gamma)=-\dfrac{\sum_{i=1}^n p_i}{\sum_{i=1}^n \zeta_i^{-1}}\,.    
\end{align}
\end{thm}
\begin{proof}
The equilibrium analysis is similar to the one presented in Section~\ref{ssec:goal}, except that the effect of the integral variables $\boldsymbol{s}$ introduced by the generalized DAI should be considered. Hence, in steady-state, \eqref{eq:sys-dyn-vec-distribute} yields
\begin{subequations}
\begin{align}\label{eq:equli-dai}
 \mathbbold{0}_n =& -
    \boldsymbol{A} \mathbbold{1}_n\omega^\ast- \nabla U (\boldsymbol{\delta}^\ast) + \boldsymbol{p} + \boldsymbol{u}(\boldsymbol{s}^\ast)\,,\\
    2 \pi f_0 \mathbbold{1}_n\omega^\ast =& -\boldsymbol{Z}\boldsymbol{L}_\mathrm{Q}\nabla C(\boldsymbol{u}(\boldsymbol{s}^\ast))\,,\label{eq:equli-dai-s}
    \end{align}
\end{subequations}
where we have used the equilibria characterizations $\boldsymbol{\omega}^\ast=\mathbbold{1}_n\omega^\ast$ and \eqref{eq:equli-sync-omega}.

We then proceed to investigate the equilibrium by following a similar argument as in~\cite[Lemma~4.2]{Johannes2016ecc}. Premultiplying \eqref{eq:equli-dai-s} by $\mathbbold{1}_n^T\boldsymbol{Z}^{-1}$ yields 
\begin{align}
2 \pi f_0 \mathbbold{1}_n^T\boldsymbol{Z}^{-1}\mathbbold{1}_n\omega^\ast =& - \mathbbold{1}_n^T\boldsymbol{L}_\mathrm{Q}\nabla C(\boldsymbol{u}(\boldsymbol{s}^\ast)) = 0\,, \label{eq:omegastar-eq}
\end{align}
where the second equality is due to the property of the Laplacian matrix~\cite{FB-LNS} that $\mathbbold{1}_n^T\boldsymbol{L}_\mathrm{Q}=\mathbbold{0}_n^T$. Note that $\mathbbold{1}_n^T\boldsymbol{Z}^{-1}\mathbbold{1}_n>0$ since $\boldsymbol{Z}^{-1}\succ 0$ by construction. Thus, \eqref{eq:omegastar-eq} implies that $\omega^\ast=0$. Then, \eqref{eq:equli-dai-s} becomes $\boldsymbol{Z}\boldsymbol{L}_\mathrm{Q}\nabla C(\boldsymbol{u}(\boldsymbol{s}^\ast))=\mathbbold{0}_n$, which indicates that $\nabla C(\boldsymbol{u}(\boldsymbol{s}^\ast))\in\range{\mathbbold{1}_n}$, i.e., $\nabla C(\boldsymbol{u}(\boldsymbol{s}^\ast)) = \gamma\mathbbold{1}_n$ for some constant $\gamma$. Thus, $\boldsymbol{u}(\boldsymbol{s}^\ast)=\left(\nabla C_i^{-1}(\gamma),i \in \mathcal{N} \right)$. Note that $\boldsymbol{Z}\boldsymbol{u}(\boldsymbol{s}^\ast)= \left(\zeta_i\nabla C_i^{-1}(\gamma),i \in \mathcal{N} \right)=\nabla C_\mathrm{o}^{-1}(\gamma)\mathbbold{1}_n$ by Lemma~\ref{lem:inverse-id}, which further implies that $\boldsymbol{u}(\boldsymbol{s}^\ast)=\nabla C_\mathrm{o}^{-1}(\gamma)\boldsymbol{Z}^{-1}\mathbbold{1}_n$.

Now, applying $\omega^\ast=0$ and $\boldsymbol{u}(\boldsymbol{s}^\ast) = \nabla C_\mathrm{o}^{-1}(\gamma)\boldsymbol{Z}^{-1}\mathbbold{1}_n$ to \eqref{eq:equli-dai} yields
\begin{align}\label{eq:equli-dai-gamma}
 \nabla U (\boldsymbol{\delta}^\ast) = \boldsymbol{p} + \nabla C_\mathrm{o}^{-1}(\gamma)\boldsymbol{Z}^{-1}\mathbbold{1}_n\,.
    \end{align}
Then, premultiplying \eqref{eq:equli-dai-gamma} by $\mathbbold{1}_n^T$ yields the equation that determines $\gamma$ as shown in \eqref{eq:gamma},
where $\mathbbold{1}_n^T\nabla U (\boldsymbol{\delta})=0$ is used again. We can show that the solution $\gamma$ to \eqref{eq:gamma} is unique by way of contradiction. Suppose that both $\gamma$ and $\tilde{\gamma}$ satisfy \eqref{eq:gamma}, where $\gamma\neq\tilde{\gamma}$. Then,
\begin{align}\label{eq:gamma-diff}
\nabla C_\mathrm{o}^{-1}(\gamma)-\nabla C_\mathrm{o}^{-1}(\tilde{\gamma})=&\left(-\dfrac{\sum_{i=1}^n p_i}{\sum_{i=1}^n \zeta_i^{-1}}\right)-\left(-\dfrac{\sum_{i=1}^n p_i}{\sum_{i=1}^n \zeta_i^{-1}}\right)\nonumber\\=&\ 0\,.  \end{align}
Multiplying \eqref{eq:gamma-diff} by $\left(\gamma-\tilde{\gamma}\right)$ yields
\begin{align}\label{eq:gamma-diff-prod}
0=& \left(\gamma-\tilde{\gamma}\right)\left(\nabla C_\mathrm{o}^{-1}(\gamma)-\nabla C_\mathrm{o}^{-1}(\tilde{\gamma})\right)>0\,,
\end{align}
where the inequality is due to the fact that $\nabla C_\mathrm{o}^{-1}(\cdot)$ is strictly increasing as discussed in Remark~\ref{rem:inverse}. Clearly, \eqref{eq:gamma-diff-prod} is a contradiction. Thus, $\gamma$ is unique.

It remains to show that the equilibrium point is unique. We focus on the uniqueness of $\boldsymbol{s}^\ast$. By way of contradiction, suppose that both $\boldsymbol{s}^\ast$ and $\boldsymbol{s}^\star$ satisfy \eqref{eq:equ-DAI-u}, where $\boldsymbol{s}^\ast\neq\boldsymbol{s}^\star$. Then,
\begin{align}\label{eq:u-diff}
\boldsymbol{u}(\boldsymbol{s}^\ast) - \boldsymbol{u}(\boldsymbol{s}^\star)=& \nabla C_\mathrm{o}^{-1}(\gamma)\boldsymbol{Z}^{-1}\mathbbold{1}_n-\nabla C_\mathrm{o}^{-1}(\gamma)\boldsymbol{Z}^{-1}\mathbbold{1}_n\nonumber\\=&\ \mathbbold{0}_n \,. 
\end{align}
Premultiplying \eqref{eq:u-diff} by $\left(\boldsymbol{s}^\ast-\boldsymbol{s}^\star\right)^T$ yields
\begin{align}\label{eq:u-diff-prod}
0=&\left(\boldsymbol{s}^\ast-\boldsymbol{s}^\star\right)^T\left(\boldsymbol{u}(\boldsymbol{s}^\ast) - \boldsymbol{u}(\boldsymbol{s}^\star)\right)\nonumber\\=&\sum_{i=1}^n \left(s_i^\ast-s_i^\star\right)\left(u_i(s_i^\ast)-u_i(s_i^\star)\right)>0 \,,  
\end{align}
where the inequality results from our requirement that $u_i(s_i)$ is strictly increasing with respect to $s_i$. Clearly, \eqref{eq:u-diff-prod} is a contradiction. Thus, $\boldsymbol{s}^\ast$ is unique. By~\cite[Lemma~1]{weitenberg2017exponentialarxiv}, the solution $\boldsymbol{\delta}^\ast$ to \eqref{eq:equli-dai-gamma} is unique. This concludes the proof of the uniqueness of the equilibrium.
\end{proof}

Theorem~\ref{thm:equilibrium} verifies that the generalized DAI preserves the steady-state performance of the normal DAI. It forces the closed-loop system \eqref{eq:sys-dyn-vec-distribute} to settle down at a unique equilibrium point where the frequency is nominal, i.e., $\boldsymbol{\omega}^\ast=\mathbbold{0}_n$, and the controllable power injections $\boldsymbol{u}^\ast$ meet the identical marginal cost requirement \eqref{eq:id-marginal} since, $\forall i\in\mathcal{N}$, 
\begin{align*}
\nabla C_i (u_i^\ast) =&\nabla C_i (u_i(s_i^\ast))\stackrel{\text{\eqref{eq:equ-DAI-u}}}{=}\nabla C_i (\nabla C_\mathrm{o}^{-1}(\gamma)\zeta_i^{-1})\\\stackrel{\text{\eqref{eq:inverse-id}}}{=}&\nabla C_i (\zeta_i\nabla C_i^{-1}(\gamma)\zeta_i^{-1})=\nabla C_i (\nabla C_i^{-1}(\gamma)) =\gamma \,.   
\end{align*}
Thus, the objectives of nominal steady-state frequency and optimal steady-state economic dispatch for a broader range of cost functions are both achieved. 

\subsection{Lyapunov Stability Analysis}
Having characterized the equilibrium point and confirmed the steady-state performance of the closed-loop system \eqref{eq:sys-dyn-vec-distribute} under the generalized DAI, we are now ready to investigate the system stability by performing Lyapunov stability analysis. More precisely, the stability under the generalized DAI can be certified by finding a well-defined ``weak'' Lyapunov function
%
%
that is nonincreasing along the trajectories of the closed-loop system~\eqref{eq:sys-dyn-vec-distribute}. The main result of this whole subsection is presented below, whose proof is enabled by a sequence of smaller results that we discuss next.

\begin{thm}[Asymptotic stability]\label{thm:as-stable}
Under Assumptions~\ref{ass:nonlinear-u} and~\ref{ass-cost-grad-scale}, any trajectory of \eqref{eq:sys-dyn-vec-distribute} that starts from the neighborhood $\Set*{(\boldsymbol{\delta}, \boldsymbol{\omega}, \boldsymbol{s})\!\in \!\real^n \!\!\times\! \real^{n}\!\!\times\! \real^n}{|\delta_i-\delta_j|\in [0,\pi/2), \forall \{i,j\}\!\in \!\mathcal{E}}$
of the equilibrium characterized by \eqref{eq:equ-DAI} with $|\delta_i^\ast-\delta_j^\ast|\in [0,\pi/2)$, $\forall \{i,j\} \in\mathcal{E}$, converges asymptotically to the equilibrium.
\end{thm}

%
%

%
%

First note that the algebraic equation \eqref{eq:sys-dyn-vec-distribute-omegaL} fully determines $\boldsymbol{\omega}_\mathcal{L}$ as a function of $\boldsymbol{\delta}$ and $\boldsymbol{s}$, i.e.,
$$ \boldsymbol{\omega}_\mathcal{L} = \boldsymbol{A}_\mathcal{L}^{-1}\left(- \nabla_\mathcal{L} U (\boldsymbol{\delta}) + \boldsymbol{p}_\mathcal{L} + \boldsymbol{u}_\mathcal{L}(\boldsymbol{s})\right)\,,$$
which means that we only need to explicitly consider the evolution of $(\boldsymbol{\delta}, \boldsymbol{\omega}_\mathcal{G}, \boldsymbol{s})$. Thus,
in what follows, we only focus on the state component $(\boldsymbol{\delta}, \boldsymbol{\omega}_\mathcal{G}, \boldsymbol{s})$ for a neighborhood $\mathcal{D}:=\! \Set*{(\boldsymbol{\delta}, \boldsymbol{\omega}_\mathcal{G}, \boldsymbol{s})\!\in \!\real^n \!\!\times\! \real^{|\mathcal{G}|}\!\!\times\! \real^n}{|\delta_i-\delta_j|\in [0,\pi/2), \forall \{i,j\}\!\in \!\mathcal{E}}$ of the equilibrium $(\boldsymbol{\delta}^\ast, \mathbbold{0}_{|\mathcal{G}|}, \boldsymbol{s}^\ast)$.
The following result formalizes this thought by showing that the distance from the whole state $(\boldsymbol{\delta}, \boldsymbol{\omega}, \boldsymbol{s})$ to $(\boldsymbol{\delta}^\ast, \mathbbold{0}_n, \boldsymbol{s}^\ast)$ is lower and upper bounded by the distance from the partial state $(\boldsymbol{\delta}, \boldsymbol{\omega}_\mathcal{G}, \boldsymbol{s})$ to $(\boldsymbol{\delta}^\ast, \mathbbold{0}_{|\mathcal{G}|}, \boldsymbol{s}^\ast)$.

\begin{lem}[Bounds on whole state distance]\label{lem:bound-state}Let $\|\cdot\|_2$ denote the Euclidean norm. Under Assumptions~\ref{ass:nonlinear-u} and~\ref{ass-cost-grad-scale}, there exists some constant $\nu>0$ such that, $\forall(\boldsymbol{\delta}, \boldsymbol{\omega}_\mathcal{G},\boldsymbol{s})\in \mathcal{D}$,
\begin{align*}
&\ \|\boldsymbol{\delta}-\boldsymbol{\delta}^\ast\|_2^2+\|\boldsymbol{\omega}_\mathcal{G}\|_2^2+\|\boldsymbol{s}-\boldsymbol{s}^\ast\|_2^2\\\leq&\ \|\boldsymbol{\delta}-\boldsymbol{\delta}^\ast\|_2^2+\|\boldsymbol{\omega}\|_2^2+\|\boldsymbol{s}-\boldsymbol{s}^\ast\|_2^2\\\leq&\ \nu\left(\|\boldsymbol{\delta}-\boldsymbol{\delta}^\ast\|_2^2+\|\boldsymbol{\omega}_\mathcal{G}\|_2^2+\|\boldsymbol{s}-\boldsymbol{s}^\ast\|_2^2\right)\,,   
\end{align*}
where $(\boldsymbol{\delta}^\ast, \mathbbold{0}_n, \boldsymbol{s}^\ast)$ is the equilibrium of the closed-loop system \eqref{eq:sys-dyn-vec-distribute} characterized by \eqref{eq:equ-DAI} with $|\delta_i^\ast-\delta_j^\ast|\in [0,\pi/2)$, $\forall \{i,j\} \in\mathcal{E}$.
\end{lem}
\begin{proof} 
See Appendix \ref{app:lem2-pf}.
\end{proof}

Lemma~\ref{lem:bound-state} indicates that it suffices to investigate the evolution of $(\boldsymbol{\delta}, \boldsymbol{\omega}_\mathcal{G}, \boldsymbol{s})$.
To do so, we need to first find a well-defined Lyapunov function $W(\boldsymbol{\delta}, \boldsymbol{\omega}_\mathcal{G}, \boldsymbol{s})$ such that $W(\boldsymbol{\delta}^\ast, \mathbbold{0}_{|\mathcal{G}|}, \boldsymbol{s}^\ast)=0$ and $W(\boldsymbol{\delta}, \boldsymbol{\omega}_\mathcal{G},\boldsymbol{s})>0$, $\forall(\boldsymbol{\delta}, \boldsymbol{\omega}_\mathcal{G},\boldsymbol{s})\in \mathcal{D}\setminus{(\boldsymbol{\delta}^\ast, \mathbbold{0}_{|\mathcal{G}|}, \boldsymbol{s}^\ast)}$, and then verify $\dot{W}(\boldsymbol{\delta}, \boldsymbol{\omega}_\mathcal{G}, \boldsymbol{s})\leq0$, $\forall(\boldsymbol{\delta}, \boldsymbol{\omega}_\mathcal{G},\boldsymbol{s})\in \mathcal{D}$.

Following the same method in \cite{DORFLER2017Auto}, we begin our construction of a Lyapunov function by defining the following integral function:
\begin{align}\label{eq:L(s)}
L(\boldsymbol{s}):=\sum_{i=1}^n \int_0^{s_i}u_i(\xi)\ \mathrm{d}\xi\,,    
\end{align}
which clearly satisfies $\nabla L (\boldsymbol{s})=\boldsymbol{u}(\boldsymbol{s})$. Another useful property related to $L(\boldsymbol{s})$ is given in the next lemma. 
\begin{lem}[Strict convexity of $L(\boldsymbol{s})$]\label{lem:strict-conv-L}
If Assumption~\ref{ass:nonlinear-u} holds, then the function $L(\boldsymbol{s})$ defined in \eqref{eq:L(s)} is strictly convex.
\end{lem}
\begin{proof} 
See the Appendix \ref{app:lem3-pf}.
\end{proof}

Having characterized the properties of $L(\boldsymbol{s})$, we consider the following Lyapunov function candidate:
\begin{align}\label{eq: Lyapunov}
W(\boldsymbol{\delta}, \boldsymbol{\omega}_\mathcal{G}, \boldsymbol{s}):=&\ \pi f_0 \boldsymbol{\omega}_\mathcal{G}^T\boldsymbol{M}\boldsymbol{\omega}_\mathcal{G} \\&+ U(\boldsymbol{\delta})-U(\boldsymbol{\delta^\ast})-\nabla U (\boldsymbol{\delta^\ast})^T\left(\boldsymbol{\delta}-\boldsymbol{\delta^\ast}\right)\nonumber\\&+ L(\boldsymbol{s})-L(\boldsymbol{s^\ast})-\nabla L (\boldsymbol{s^\ast})^T\left(\boldsymbol{s}-\boldsymbol{s^\ast}\right)\,,\nonumber    \end{align}
where $(\boldsymbol{\delta}^\ast, \mathbbold{0}_{|\mathcal{G}|}, \boldsymbol{s}^\ast)$ corresponds to the unique equilibrium point of the closed-loop system \eqref{eq:sys-dyn-vec-distribute} satisfying \eqref{eq:equ-DAI}. The next result shows that this is  a well-defined candidate Lyapunov function on~$\mathcal{D}$.

\begin{lem}[Well-defined Lyapunov function]\label{lem:W}
Let Assumptions~\ref{ass:nonlinear-u} and~\ref{ass-cost-grad-scale} hold. The function $W(\boldsymbol{\delta}, \boldsymbol{\omega}_\mathcal{G}, \boldsymbol{s})$ defined in \eqref{eq: Lyapunov} satisfies
$W(\boldsymbol{\delta}^\ast, \mathbbold{0}_{|\mathcal{G}|}, \boldsymbol{s}^\ast)=0$ and $W(\boldsymbol{\delta}, \boldsymbol{\omega}_\mathcal{G}, \boldsymbol{s})>0, \forall(\boldsymbol{\delta}, \boldsymbol{\omega}_\mathcal{G}, \boldsymbol{s})\in \mathcal{D}\setminus{(\boldsymbol{\delta}^\ast, \mathbbold{0}_{|\mathcal{G}|}, \boldsymbol{s}^\ast)}$.
\end{lem}
\begin{proof}
See the Appendix \ref{app:lem4-pf}.
\end{proof}

Next, we examine the system stability through the derivative of $W(\boldsymbol{\delta}, \boldsymbol{\omega}_\mathcal{G}, \boldsymbol{s})$ along the trajectories of the closed-loop system \eqref{eq:sys-dyn-vec-distribute}, whose expression is provided by the next result.

\begin{lem}[Directional derivative of Lyapunov function]\label{lem:Wdot}
Let Assumptions~\ref{ass:nonlinear-u} and~\ref{ass-cost-grad-scale} hold. Then the derivative of the Lyapunov function $W(\boldsymbol{\delta}, \boldsymbol{\omega}_\mathcal{G}, \boldsymbol{s})$ defined in \eqref{eq: Lyapunov} along the trajectories of the closed-loop system \eqref{eq:sys-dyn-vec-distribute} is given by
\begin{align}\label{eq:Wdot-distribute}
    &\dot{W}(\boldsymbol{\delta}, \boldsymbol{\omega}_\mathcal{G}, \boldsymbol{s}) = \\
    &-2\pi f_0 \boldsymbol{\omega}_\mathcal{G}^T\boldsymbol{A}_\mathcal{G} \boldsymbol{\omega}_\mathcal{G}- \boldsymbol{u}(\boldsymbol{s})^T \boldsymbol{Z}\boldsymbol{L}_\mathrm{Q}\nabla C(\boldsymbol{u}(\boldsymbol{s}))\nonumber\\
    &-2\pi f_0 \left(\nabla_\mathcal{L} U (\boldsymbol{\delta}^\ast)-\nabla_\mathcal{L} U (\boldsymbol{\delta})+\boldsymbol{u}_\mathcal{L}(\boldsymbol{s})-\boldsymbol{u}_\mathcal{L}(\boldsymbol{s}^\ast)\right)^T\boldsymbol{A}_\mathcal{L}^{-1}\nonumber\\&\quad\cdot\left(\nabla_\mathcal{L} U (\boldsymbol{\delta}^\ast)- \nabla_\mathcal{L} U (\boldsymbol{\delta})  + \boldsymbol{u}_\mathcal{L}(\boldsymbol{s})\!-\boldsymbol{u}_\mathcal{L}(\boldsymbol{s}^\ast)\right)\,.\nonumber
\end{align}
\end{lem}
\begin{proof}
See the Appendix \ref{app:lem5-pf}.
\end{proof}

Our next step is to show that $\dot{W}(\boldsymbol{\delta}, \boldsymbol{\omega}_\mathcal{G}, \boldsymbol{s})$ is nonpositive. From \eqref{eq:Wdot-distribute}, we observe that the quadratic terms defined by $\boldsymbol{A}_\mathcal{G}$ and  $\boldsymbol{A}_\mathcal{L}^{-1}$ are clearly negative. Thus, it only remains to determine the sign of the the cross-term $\boldsymbol{u}(\boldsymbol{s})^T \boldsymbol{Z}\boldsymbol{L}_\mathrm{Q}\nabla C(\boldsymbol{u}(\boldsymbol{s}))$ defined by the scaled Laplacian matrix $\boldsymbol{Z}\boldsymbol{L}_\mathrm{Q}$. 
The first thing we notice is that this term as a bilinear form with respect to the scaled Laplacian matrix $\boldsymbol{Z}\boldsymbol{L}_\mathrm{Q}$ can be expanded as the expression provided in the next result.

\begin{lem}[Bilinear form on $\real^n$ for scaled $\boldsymbol{L}_\mathrm{Q}$]\label{lem:bilinear-L}
$\forall\boldsymbol{x}:=\left(x_i, i \in \mathcal{N} \right) \in \real^n, \boldsymbol{y}:=\left(y_i, i \in \mathcal{N} \right)\in\real^n$,
\begin{align*}
\boldsymbol{x}^T \boldsymbol{Z}\boldsymbol{L}_\mathrm{Q} \boldsymbol{y}=\sum_{\{i,j\} \in\mathcal{E}_\mathrm{Q}}\!\!\!Q_{ij}\left( y_i  - y_j \right)\left(\zeta_i x_i-\zeta_j x_j\right)\,.
\end{align*}
\end{lem}
\begin{proof}
See the Appendix \ref{app:lem6-pf}.
\end{proof}
%
%

With the help of Lemma~\ref{lem:bilinear-L}, the problem of determining the sign of the cross term $\boldsymbol{u}(\boldsymbol{s})^T \boldsymbol{Z}\boldsymbol{L}_\mathrm{Q}\nabla C(\boldsymbol{u}(\boldsymbol{s}))$ can be solved using the following corollary.

%
%


\begin{cor}[Sign of cross term]\label{cor:sign-cross}
If Assumption~\ref{ass-cost-grad-scale} holds, then
\begin{align*}
    \boldsymbol{u}(\boldsymbol{s})^T \boldsymbol{Z}\boldsymbol{L}_\mathrm{Q}\nabla C(\boldsymbol{u}(\boldsymbol{s}))\geq0
\end{align*}
with equality holding if and only if when $\nabla C(\boldsymbol{u}(\boldsymbol{s}))\in\range{\mathbbold{1}_n}$.
\end{cor}
\begin{proof}
See the Appendix \ref{app:cor2-pf}.
\end{proof}

We now have all the elements necessary to establish the stability of the equilibrium as summarized in Theorem~\ref{thm:as-stable} by using $W$ as a Lyapunov function. Therefore, we are ready to finish the proof of Theorem~\ref{thm:as-stable}. Using the expression of $\dot{W}(\boldsymbol{\delta}, \boldsymbol{\omega}_\mathcal{G}, \boldsymbol{s})$ in~\eqref{eq:Wdot-distribute},
it follows directly from the fact that $\boldsymbol{A}_\mathcal{G}\succ0$, $\boldsymbol{A}_\mathcal{L}^{-1}\succ0$, and $\boldsymbol{u}(\boldsymbol{s})^T \boldsymbol{Z}\boldsymbol{L}_\mathrm{Q}\nabla C(\boldsymbol{u}(\boldsymbol{s}))\geq0$ by Corollary~\ref{cor:sign-cross} that $\dot{W}(\boldsymbol{\delta}, \boldsymbol{\omega}_\mathcal{G}, \boldsymbol{s})\leq0$. Observe from \eqref{eq:Wdot-distribute} that $\dot{W}(\boldsymbol{\delta}, \boldsymbol{\omega}_\mathcal{G}, \boldsymbol{s})\equiv0$ directly enforces $\boldsymbol{\omega}_\mathcal{G}\equiv\mathbbold{0}_{|\mathcal{G}|}$, $\boldsymbol{u}(\boldsymbol{s})^T \boldsymbol{Z}\boldsymbol{L}_\mathrm{Q}\nabla C(\boldsymbol{u}(\boldsymbol{s}))\equiv0$, and $\nabla_\mathcal{L} U (\boldsymbol{\delta})-\boldsymbol{u}_\mathcal{L}(\boldsymbol{s})\equiv\nabla_\mathcal{L} U (\boldsymbol{\delta}^\ast)-\boldsymbol{u}_\mathcal{L}(\boldsymbol{s}^\ast)=\boldsymbol{p}_\mathcal{L}$. Clearly, $\boldsymbol{\omega}_\mathcal{G}\equiv\mathbbold{0}_{|\mathcal{G}|}$ implies that $\dot{\boldsymbol{\omega}}_\mathcal{G}\equiv\mathbbold{0}_{|\mathcal{G}|}$. By \eqref{eq:sys-dyn-vec-distribute-omegaL}, $\nabla_\mathcal{L} U (\boldsymbol{\delta})-\boldsymbol{u}_\mathcal{L}(\boldsymbol{s})\equiv\boldsymbol{p}_\mathcal{L}$ ensures that $\boldsymbol{\omega}_\mathcal{L}\equiv\mathbbold{0}_{|\mathcal{L}|}$. It follows from $\boldsymbol{\omega}_\mathcal{G}\equiv\mathbbold{0}_{|\mathcal{G}|}$ and $\boldsymbol{\omega}_\mathcal{L}\equiv\mathbbold{0}_{|\mathcal{L}|}$ that $\boldsymbol{\omega}\equiv\mathbbold{0}_{n}$, which further indicates that $\dot{\boldsymbol{\delta}}\equiv\mathbbold{0}_n$ by \eqref{eq:sys-dyn-vec-distribute-delta}. Finally, by Corollary~\ref{cor:sign-cross}, $\boldsymbol{u}(\boldsymbol{s})^T \boldsymbol{Z}\boldsymbol{L}_\mathrm{Q}\nabla C(\boldsymbol{u}(\boldsymbol{s}))\equiv0$ is equivalent to $\nabla C(\boldsymbol{u}(\boldsymbol{s}))\in\range{\mathbbold{1}_n}$ constantly. Hence, $\boldsymbol{Z}\boldsymbol{L}_\mathrm{Q}\nabla C(\boldsymbol{u}(\boldsymbol{s}))\equiv\mathbbold{0}_{n}$, which together with $\boldsymbol{\omega}\equiv\mathbbold{0}_{n}$ implies that $\dot{\boldsymbol{s}}\equiv\mathbbold{0}_n$ based on \eqref{eq:bus-dyn}. Thus, we have shown that $\dot{W}(\boldsymbol{\delta}, \boldsymbol{\omega}_\mathcal{G}, \boldsymbol{s})\equiv0$ implies $\dot{\boldsymbol{\delta}}\equiv\mathbbold{0}_n$, $\dot{\boldsymbol{\omega}}_{\mathcal{G}}\equiv\mathbbold{0}_{|\mathcal{G}|}$, and $\dot{\boldsymbol{s}}\equiv\mathbbold{0}_n$. This means that the largest invariant set contained in the set of points where $\dot{W}$ vanishes is actually the set of equilibria. Therefore, by the LaSalle invariance principle~\cite[Theorem~4.4]{khalil2002nonlinear}, every trajectory of the closed-loop system \eqref{eq:sys-dyn-vec-distribute} starting within $\mathcal{D}$ converges to the equilibrium set, which by Theorem~\ref{thm:equilibrium} consists of a unique  point satisfying \eqref{eq:equ-DAI}. This concludes the proof of local asymptotic stability.

Theorem~\ref{thm:as-stable} shows that the closed-loop system \eqref{eq:sys-dyn-vec-distribute} under the generalized DAI is locally asymptotically stabilized to the unique equilibrium point characterized by Theorem~\ref{thm:equilibrium}, where the steady-state performance objectives are achieved even for a broader range of cost functions beyond quadratic ones.

%% file: RL.tex
In the previous section, we showed that any controller $\boldsymbol{u}$ that is monotonic and through the origin would drive the system to the optimal steady-state solution. However, different controllers in this class may lead to very different transient behaviors. In this section, we focus on integrating reinforcement learning (RL) into our proposed generalized DAI to design a controller that optimizes the transient performance of the system without jeopardizing its stability. Basically, after parameterizing the control policy $u_i(s_i)$ that maps the integral variable $s_i$ of the generalized DAI to the controllable power injection $u_i$ on each bus as a monotonic neural network, we train those neural networks by a RL algorithm based on a recurrent neural network (RNN)~\cite{cui2022tps}. 

\subsection{Monotonic Neural Networks for Stability Guarantee}

Recall from Section~\ref{sec:Generalized-DAI} that the stability of the closed-loop system \eqref{eq:sys-dyn-vec-distribute} under the generalized DAI is ensured by any nonlinear control policy $u_i(s_i)$ that is a Lipschitz continuous and strictly increasing function of $s_i$ with $u_i(0)=0$ provided that the cost of each buses satisfy Assumption~\ref{ass-cost-grad-scale}. Within this class of stabilizing controllers, we want to find one that has the best transient performance. That is, we want to minimize the frequency excursions and control efforts when the frequencies are recovering to their nominal values. 

In principle, optimizing all controllers satisfying Assumption~\ref{ass:nonlinear-u} is an infinite-dimensional problem. To make the problem tractable, we need to parameterize $u_i(s_i)$ in some way. Neural networks emerge as as a natural candidate due to their potential for universal approximation~\cite{zhouNIPS2017express}. The key requirement we impose on any neural network used is that $u_i$ should be an increasing function that goes through the origin. 

Here we use the result in~\cite[Theorem~2]{cui2022tps}, which states any function of our interest can be approximated accurately enough by a single hidden layer fully-connected neural network with rectified linear unit (ReLU) activation functions\footnote{ReLU is a widely used activation function which applies  $\max \left(0,x\right)$ to any element $x$ of its argument.} for properly designed weights and biases of the neural network. Such a neural network is called a \emph{stacked ReLU monotonic neural network} whose architecture is provided in Fig.~\ref{fig:stack-relu-diagram}.


As illustrated in Fig.~\ref{fig:stack-relu-diagram}, each stacked ReLU monotonic neural network is essentially a parameterized mapping from the state $s_i$ to the control $u_i$ with parameters $\boldsymbol{k}^{+}_{i}:=\left(k^+_{i,j}, j \in \mathcal{D} \right)$, $\boldsymbol{b}^{+}_{i}:=\left(b^+_{i,j}, j \in \mathcal{D} \right)$, $\boldsymbol{k}^{-}_{i}:=\left(k^-_{i,j}, j \in \mathcal{D} \right)$, and $\boldsymbol{b}^{-}_{i}:=\left(b^-_{i,j}, j \in \mathcal{D} \right)$, where $\mathcal{D}:=\{1,\dots, d\}$ and $2d$ is the number of neurons in the hidden layer. Here, the vectors $\boldsymbol{k}^{+}_{i}$ and $\boldsymbol{k}^{-}_{i}$ are referred to as the ``weight'' vectors and the vectors $\boldsymbol{b}^{+}_{i}$ and $\boldsymbol{b}^{-}_{i}$ are referred to as the ``bias'' vectors. 

\begin{figure}[ht]
\centering
\includegraphics[width=0.6\columnwidth]{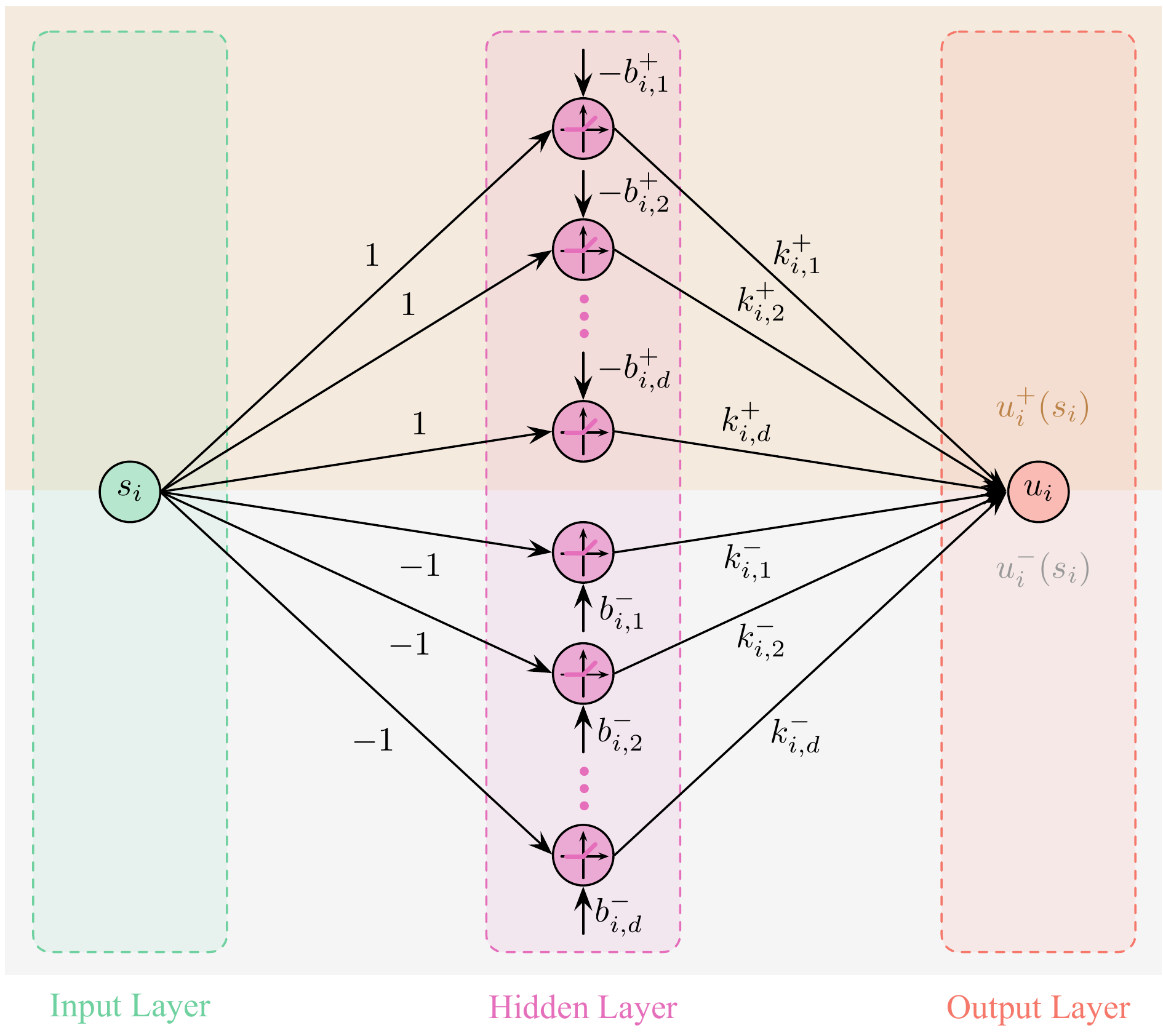}
\caption{Stacked ReLU monotonic neural network.}
\label{fig:stack-relu-diagram}
\end{figure}

These weights and biases restrict the shape of $u_i(s_i)$. The output of the top $d$ neurons in the hidden layer is zero to $s_i \leq 0$. Similarly, the bottom $d$ neurons zeros out any positive inputs. Therefore, we write $u_i(s_i)=u_i^+(s_i)+u_i^-(s_i)$, where $u_i^+(s_i)$ corresponds to the top $d$ neurons and operate if the input is positive, and $u_i^-(s_i)$ corresponds to the bottom $d$ neurons and operate if the input is negative. 

The input layer transforms the input $s_i$ linearly as $\left(s_i-b^+_{i,j}\right)$ and then processed by the ReLU in the $j$th hidden neuron, $\forall j \in \mathcal{D}$.  Denote ReLU operator as $\sigma(\cdot)$ for simplicity. Then, the output of the $j$th hidden neuron in the top part becomes $\sigma(s_i-b^+_{i,j})=\max \left(0,s_i-b^+_{i,j}\right)$, which is linearly weighted by $k^+_{i,j}$ before contributing to the final output. Therefore, the contribution from the $j$th hidden neuron to the final output is $k^+_{i,j}\sigma(s_i-b^+_{i,j})= k^+_{i,j}\max \left(0,s_i-b^+_{i,j}\right)$. Therefore, the $j$th hidden neuron is activated only if the state $s_i$ exceeds the threshold $b^+_{i,j}$, as illustrated in Fig.~\ref{fig:principle-stack}. At the output layer, the contributions from all $d$ hidden neurons in the top part yield  
\begin{align*}
u^{+}_i(s_i):=\sum_{j=1}^dk^+_{i,j}\sigma(s_i-b^+_{i,j})=\left(\boldsymbol{k}^{+}_{i}\right)^T\sigma\left(\mathbbold{1}_d s_i-\boldsymbol{b}^+_{i}\right)\,.    
\end{align*}
The negative part $u_i^-$ can be treated the same way. 


Of course, for $u_i$ to be increasing, the weights and biases need to be constrained. We do this by making the slope always positive, as shown in Fig.~\ref{fig:control-curve}. This gives rise to the following requirements on weights and biases:  
\begin{subequations}\label{eq:cons-bk}
\begin{align}
    b^{-}_{i,d}\leq \dots \leq b^{-}_{i,2}\leq b^{-}_{i,1}\!=0&=b^{+}_{i,1}\leq b^{+}_{i,2}\leq\!\dots\leq b^{+}_{i,d}\,,\label{eq:bias}\\-\infty<\sum_{j=1}^{\iota} k^-_{i,j}<0&<\sum_{j=1}^{\iota} k^+_{i,j}<\infty\,,\ \forall \iota \in\! \mathcal{D}\,.\label{eq:weights}
\end{align}
\end{subequations}
If these are satisfied, the stacked ReLU monotonic neural network is a control policy $u_i(s_i)$ satisfying Assumption~\ref{ass:nonlinear-u}.

\begin{figure}[t!]
\centering
\subfigure[Monotonic Control policy]
{\includegraphics[width=1\columnwidth]{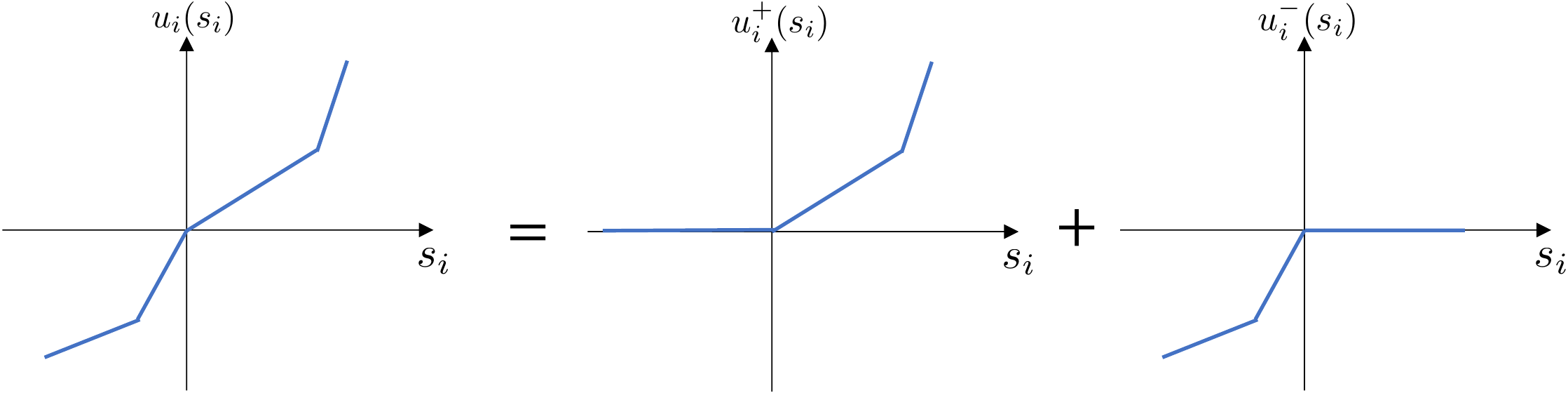}\label{fig:control-curve}}
\hfil
\subfigure[Principle of stacking hidden neurons]
{\includegraphics[width=1\columnwidth]{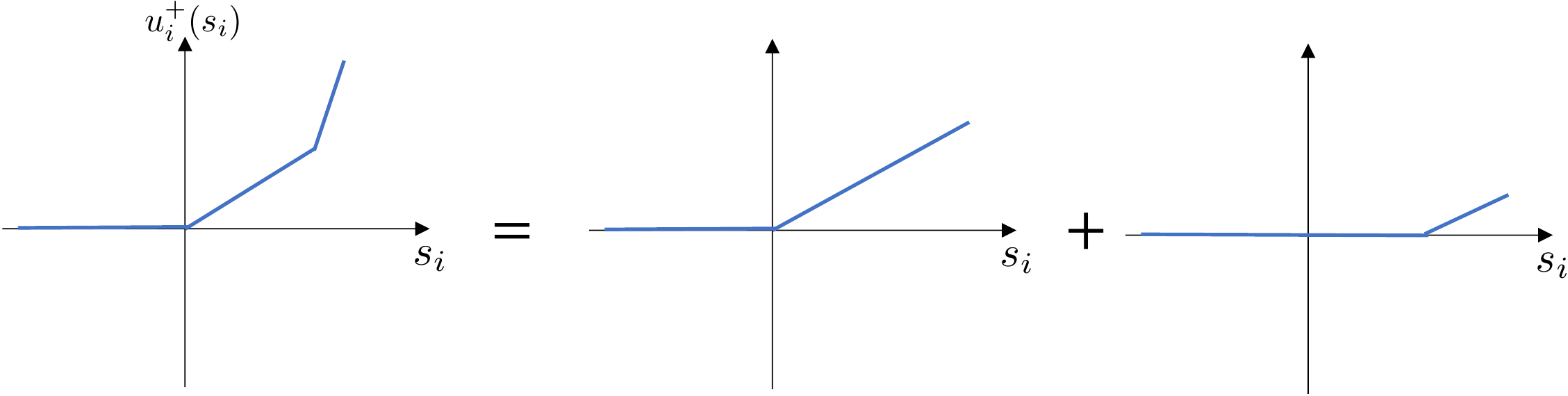}\label{fig:principle-stack}}
\caption{Sketch of how hidden neurons in the $i$th stacked ReLU monotonic neural network help to form a control policy $u_i(s_i)$.}
\label{fig:piece-linear}
\end{figure}

The inequality constraints in \eqref{eq:cons-bk} are not trivial to enforce when training the neural network. It turns out there is a simple trick to address this challenge. We introduce a group of intermediate parameters $\boldsymbol{\mu}^{+}_{i}:=\left(\mu^+_{i,j}, j \in \mathcal{D} \right)\in \real^d$, $\boldsymbol{\chi}^{+}_{i}:=\left(\chi^+_{i,j}, j \in \mathcal{D}\setminus{\{d\}} \right)\in \real^{d-1}$, $\boldsymbol{\mu}^{-}_{i}:=\left(\mu^-_{i,j}, j \in \mathcal{D} \right)\in \real^d$, and $\boldsymbol{\chi}^{-}_{i}:=\left(\chi^-_{i,j}, j \in \mathcal{D}\setminus{\{d\}} \right)\in \real^{d-1}$ that are unconstrained. The original parameters $\boldsymbol{k}^{+}_{i}$, $\boldsymbol{b}^{+}_{i}$, $\boldsymbol{k}^{-}_{i}$, and $\boldsymbol{b}^{-}_{i}$ are given by
\begin{subequations}
\begin{align}
\!\!\!\!&k^+_{i,1}\!=\!\left(\mu^+_{i,1}\right)^2\!\!\!\,,\  k^+_{i,j}\!=\!\left(\mu^+_{i,j}\right)^2\!\!-\!\left(\mu^+_{i,j-1}\right)^2\!\!\,,\ \!\forall j \!\in\! \mathcal{D}\!\setminus{\!\{1\}}\,,\\
\!\!\!\!&k^-_{i,1}\!=-\left(\mu^-_{i,1}\right)^2\,,\  k^-_{i,j}=-\left(\mu^-_{i,j}\right)^2+\left(\mu^-_{i,j-1}\right)^2\,,\\
\!\!\!\!&b^+_{i,1}\!=b^-_{i,1}\!=0\,,\ b^+_{i,j}\!=\!\!\sum_{l=1}^{j-1}\!\left(\chi^+_{i,l}\right)^2\!\!\!\,,\ b^-_{i,j}\!=\!-\!\!\sum_{l=1}^{j-1}\!\left(\chi^-_{i,l}\right)^2\!\!\!\,.
\end{align}
\end{subequations}
In this way, the requirements in \eqref{eq:cons-bk} naturally hold. Thus, the remaining problem is just the search of optimal parameters $\boldsymbol{\mu}^{+}_{i}$, $\boldsymbol{\chi}^{+}_{i}$, $\boldsymbol{\mu}^{-}_{i}$, and $\boldsymbol{\chi}^{-}_{i}$ through the training of neural networks. To highlight the dependence of the control policy $u_i(s_i)$ generated by the $i$th neural network on parameters $\boldsymbol{\mu}^{+}_{i}$, $\boldsymbol{\chi}^{+}_{i}$, $\boldsymbol{\mu}^{-}_{i}$, and $\boldsymbol{\chi}^{-}_{i}$, we will use the notation $u_i(s_i;\Xi_i(d))$ in the rest of this manuscript, where $\Xi_i(d)$ denote a set of parameters $\Xi_i:=\{\boldsymbol{\mu}^{+}_{i},\boldsymbol{\chi}^{+}_{i},\boldsymbol{\mu}^{-}_{i},\boldsymbol{\chi}^{-}_{i}\}$ associated with a stacked ReLU monotonic neural network that has $2d$ hidden neurons.

\subsection{Recurrent Neural Network-Based Reinforcement Learning}\label{ssec:RNN-train}
A recurrent neural network (RNN) is a class of neural networks that have a looping mechanism to allow information learned from the previous step to flow to the next step.
Thus, to efficiently train the stacked ReLU monotonic neural networks for parameters $\Xi_i(d)$, $\forall i \in\! \mathcal{N}$,  we adopt the RNN-based RL algorithm proposed by~\cite{cui2022tps}, whose scheme is shown in Fig.~\ref{fig:RL}. 

To fit into the RNN architecture, we have to discretize the problem \eqref{eq:opt-ts}. Let the $l$th sampling instant after the initial time $t_0$ be $(t_0 + lh)$, where $h$ is the sampling interval that is small enough. Then, using voltage angles $\boldsymbol{\theta}:=\left(\theta_i, i \in \mathcal{N} \right) \in \real^n$ as an example, we would like to approximate the continuous-time state at the $l$th sampling instant $\boldsymbol{\theta}(t_0 + lh)$ as a discrete-time state $\boldsymbol{\theta}^{\langle l\rangle}$. Provided that $\boldsymbol{\theta}^{\langle 0\rangle} = \boldsymbol{\theta}(t_0)$, the simplest way to obtain $\boldsymbol{\theta}^{\langle l\rangle} \approx \boldsymbol{\theta}(t_0 + lh)$ is through the reccurence relation via Euler method, i.e., $\boldsymbol{\theta}^{\langle l+1\rangle} \!= \boldsymbol{\theta}^{\langle l\rangle} + 2 \pi f_0 h \boldsymbol{\omega}^{\langle l\rangle}$. We solve the discretized problem below instead:   
\begin{eqnarray}
		\label{eq:opt-ts-obj-dt}
		\min_{\Xi_i(d), i \in \mathcal{N}} \!\!\!\!\!\!\!&& \sum_{i \in \mathcal{G}} \|\omega_i^{\langle l\rangle}\|_\infty+\rho\sum_{i=1}^n\tilde{C}_{i,T}  \label{eq:opt-ts-dis} \\
		\label{eq:opt-ts-con-dis}
		\mathrm{s.t.} \!\!\!\!\!\!\!\!&& \boldsymbol{\omega}_\mathcal{G}^{\langle l+1\rangle} \!\!=\! \left(I_{|\mathcal{G}|}-
    h\boldsymbol{M}^{-1}\boldsymbol{A}_\mathcal{G}\right) \boldsymbol{\omega}_\mathcal{G}^{\langle l\rangle}\nonumber\\&&\qquad\qquad\!\!+h\boldsymbol{M}^{-1}\!\left(\! - \nabla_\mathcal{G} U (\boldsymbol{\theta}^{\langle l\rangle})\!+\!\boldsymbol{p}_\mathcal{G}\!+\! \boldsymbol{u}_\mathcal{G}(\boldsymbol{s}^{\langle l\rangle})\right)\!\!\,,\nonumber\\
    && \boldsymbol{\omega}_\mathcal{L}^{\langle l\rangle} = \ \boldsymbol{A}_\mathcal{L}^{-1}\left(- \nabla_\mathcal{L} U (\boldsymbol{\theta}^{\langle l\rangle}) + \boldsymbol{p}_\mathcal{L} + \boldsymbol{u}_\mathcal{L}(\boldsymbol{s}^{\langle l\rangle})\right)\,,\nonumber\\
    &&  
    \boldsymbol{\theta}^{\langle l+1\rangle} \!= \boldsymbol{\theta}^{\langle l\rangle} + 2 \pi f_0 h \boldsymbol{\omega}^{\langle l\rangle}\,,\nonumber\\
    &&\boldsymbol{s}^{\langle l+1\rangle} \!= \boldsymbol{s}^{\langle l\rangle}\! -\!h\!\left(2 \pi f_0 \boldsymbol{\omega}^{\langle l\rangle}\! +\! \boldsymbol{Z}\boldsymbol{L}_\mathrm{Q}\nabla C(\boldsymbol{u}(\boldsymbol{s}^{\langle l\rangle}))\right)\!\!\,,\nonumber\\&&\boldsymbol{u}(\boldsymbol{s}^{\langle l\rangle}):=\left(u_i(s_i^{\langle l\rangle};\Xi_i(d)), i \in \mathcal{N} \right)\,,\nonumber
	\end{eqnarray}
with
\begin{subequations}\label{eq:compute-loss}
\begin{align}
    &\|\omega_i^{\langle l\rangle}\|_\infty:=\max_{l\in\{0,\dots, \lfloor T/h\rfloor-1\}} |\omega_i^{\langle l+1\rangle}|\,,\\ &\tilde{C}_{i,T}:=\dfrac{1}{\lfloor T/h\rfloor} \!\!\!\sum_{l=0}^{\lfloor T/h\rfloor-1}\!\!\!C_i(u_i(s_i^{\langle l\rangle};\Xi_i(d)))\,,
\end{align}
\end{subequations}
where $\lfloor T/h\rfloor$ denotes the floor of $T/h$.

\begin{figure}[t!]
\centering
\subfigure[RL scheme for training stacked ReLU monotonic neural networks]
{\includegraphics[width=0.75\columnwidth]{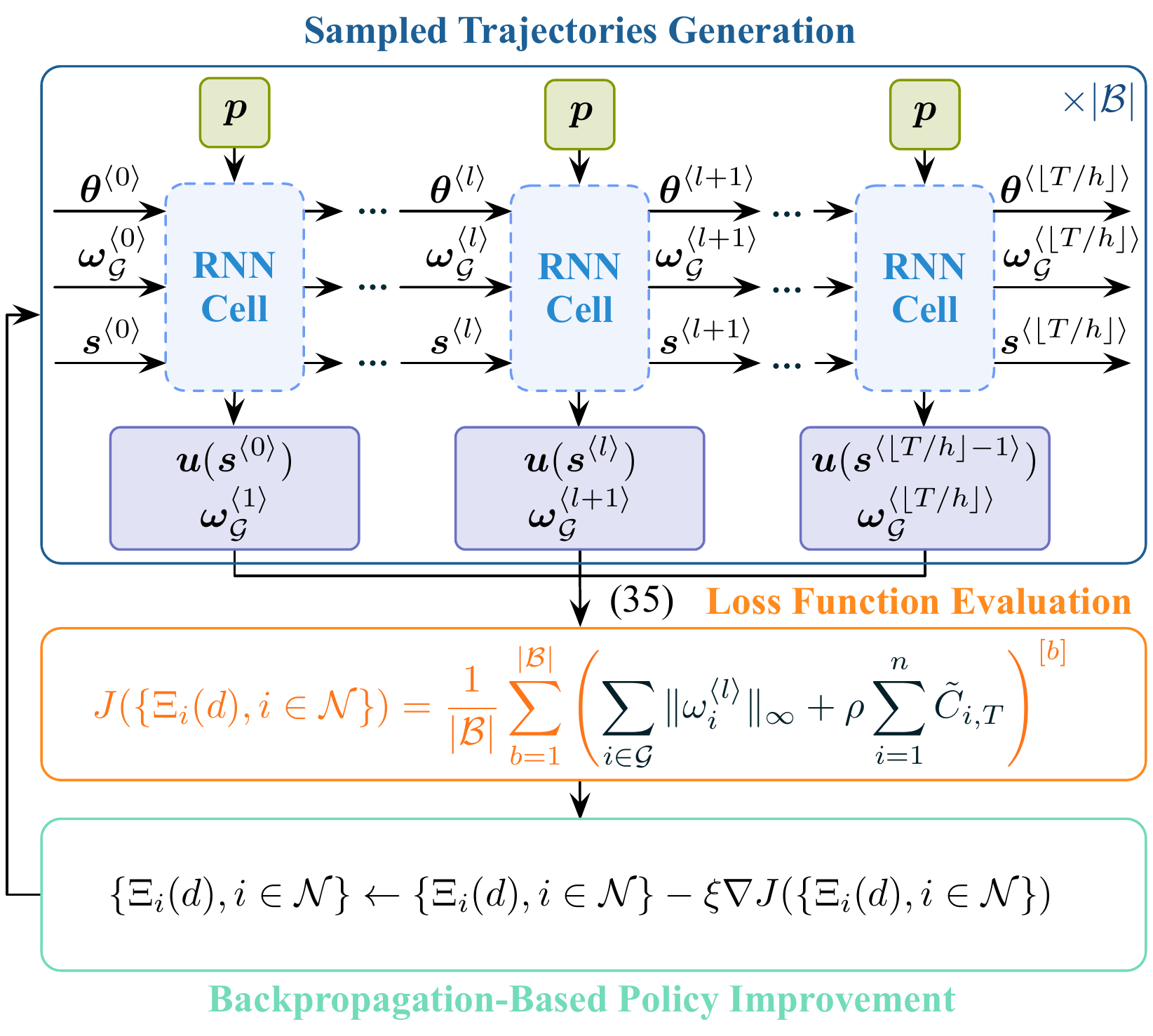}\label{fig:RL}}
\hfil
\subfigure[Structure of the recurrent neural network cell]
{\includegraphics[width=0.9\columnwidth]{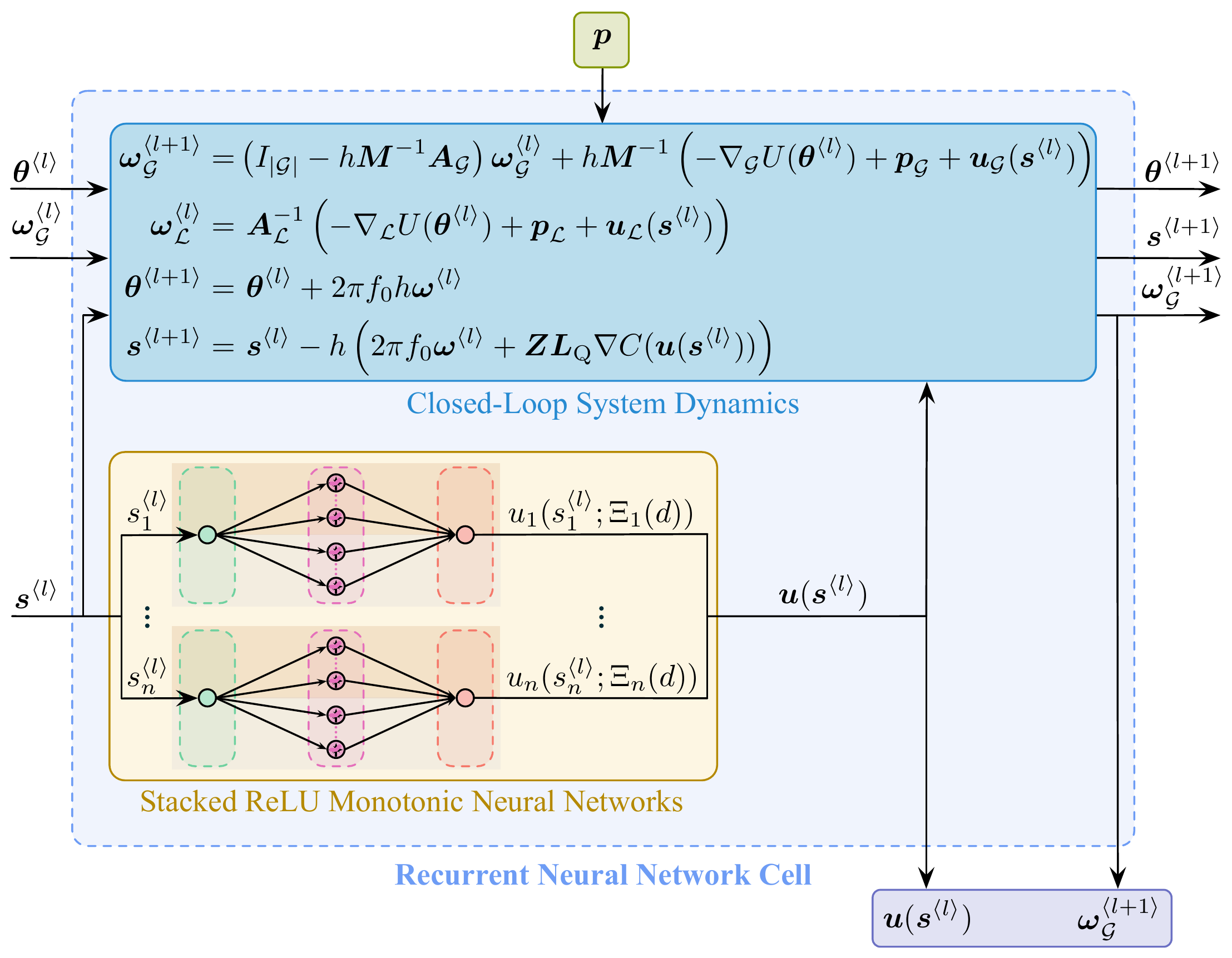}\label{fig:rnn-cell}}
\caption{Illustration of the RNN-based RL algorithm.}
\end{figure}

After a random initialization of $\{\Xi_i(d), i \in\mathcal{N}\}$, the RNN-based RL algorithm solves the problem \eqref{eq:opt-ts-obj-dt} by gradually learning the optimal $\{\Xi_i(d), i \in\mathcal{N}\}$ through an iterative process mainly composed of three parts --- sampled trajectories generation, loss function evaluation, and backpropagation-based policy improvement --- as illustrated in Fig.~\ref{fig:RL}. For a given operating point $(\boldsymbol{\theta}^{\langle 0\rangle}, \boldsymbol{\omega}_\mathcal{G}^{\langle 0\rangle}, \boldsymbol{s}^{\langle 0\rangle})$ of the power system, a sample of system trajectory over a time horizon $T$ following sudden power disturbances can be generated by randomly initializing the constant input $\boldsymbol{p}$ to the RNN cell and then simulating forward $\lfloor T/h\rfloor$ time steps. The detailed structure of the RNN cell is shown in Fig.~\ref{fig:rnn-cell}, where the control signals $\boldsymbol{u}(\boldsymbol{s}^{\langle l\rangle})$ produced by $n$ stacked ReLU monotonic neural networks with parameters $\{\Xi_i(d), i \in\mathcal{N}\}$ are applied to the discretized closed-loop system with states $(\boldsymbol{\theta}^{\langle l\rangle}, \boldsymbol{\omega}_\mathcal{G}^{\langle l\rangle}, \boldsymbol{s}^{\langle l\rangle})$ for the generation of $(\boldsymbol{\theta}^{\langle l+1\rangle}, \boldsymbol{\omega}_\mathcal{G}^{\langle l+1\rangle}, \boldsymbol{s}^{\langle l+1\rangle})$. In this way, the RNN cell evolves forward for $l=0,\dots, \lfloor T/h\rfloor-1$, which yields one trajectory under the parameters $\{\Xi_i(d), i \in\mathcal{N}\}$ that are shared temporally. The loss function related to this trajectory is defined as the objective function of the problem \eqref{eq:opt-ts-obj-dt}, which can be computed by applying \eqref{eq:compute-loss} to the stored historical outputs of the RNN cell. Usually, to improve the efficiency of training, multiple trajectories are generated parallelly, where each of them is excited by a randomly initialized constant input $\boldsymbol{p}$. The collection of these  trajectories is called a batch and the number of such trajectories is called the batch size. Then, the loss function related to a batch of size $|\mathcal{B}|$ is defined as the mean of the losses related to individual trajectories in the batch, i.e.,
\begin{align*}
    J(\{\Xi_i(d), i \in \mathcal{N}\}):=\!\dfrac{1}{|\mathcal{B}|}\!\sum_{b=1}^{|\mathcal{B}|}\!\left(\sum_{i \in \mathcal{G}} \|\omega_i^{\langle l\rangle}\|_\infty+\rho\sum_{i=1}^n\tilde{C}_{i,T}\right)^{[b]}\!\!\!\,,
\end{align*}
where, with abuse of notation, we simply introduce a superscript $[b]$ to denote the loss along the $b$th trajectory in the batch rather than accurately distinguish $\|\omega_i^{\langle l\rangle}\|_\infty$ and $\tilde{C}_{i,T}$ along different trajectories to avoid complicating the notations too much. Once the loss of the batch $J(\{\Xi_i(d), i \in \mathcal{N}\})$ has been calculated, the parameters $\{\Xi_i(d), i \in\mathcal{N}\}$ are adjusted according to the batch gradient descent method with a suitable learning rate $\xi$, i.e.,
\begin{align*}
    \{\Xi_i(d), i \in \mathcal{N}\}\leftarrow\{\Xi_i(d), i \in \mathcal{N}\}-\xi\nabla J(\{\Xi_i(d), i \in \mathcal{N}\})\,,
\end{align*}
where the gradient is approximated efficiently through backward gradient propagation. This ends one epoch of the iterative training process. The procedure is repeated until the number of epochs $E$ is reached, which produces the optimal $\{\Xi_i(d), i \in\mathcal{N}\}$ for the stacked ReLU monotonic neural networks so that the corresponding control policy $u_i(s_i;\Xi_i(d))$ on each bus helps the generalized DAI to best handle the optimal transient frequency control problem \eqref{eq:opt-ts}.



%% file: simulation.tex
In this section, we provide numerical validation for the performance of our RL-DAI on the $39$-bus New England system~\cite{athay1979practical}. First, we will train the optimal control policy $\boldsymbol{u}(\boldsymbol{s})=\left(u_i(s_i;\Xi_i(d)), i \in \mathcal{N} \right)$ for the RL-DAI on the power system model as described in Section~\ref{ssec:RNN-train}. Then, we will validate the performance of the trained RL-DAI in response to sudden step changes in power. 



The dynamic model of the $39$-bus New
England system contains $10$ generator buses and $29$ load buses, whose union is denoted as $\mathcal{N}$. Each of the $10$ generator buses is distinctly indexed by some $i\in\left\{30,\ldots,39\right\}:=\mathcal{G}$ and each of the $29$ load buses is distinctly indexed by some $i\in\left\{1,\ldots,29\right\}:=\mathcal{L}$. The generator inertia constant $m_i$, the voltage magnitude $v_i$, and the  susceptance $B_{ij}$ are directly obtained from the dataset. Given that the values of frequency sensitivity coefficients are not provided by the dataset, we set $\alpha_i = \SI{150}{\pu}$, $\forall i\in\mathcal{G}$, and $\alpha_i= \SI{100}{\pu}$, $\forall i\in\mathcal{L}$.\footnote{All per unit values are on the system base, where the system power base
is $S_0=\SI{100}{\mega\VA}$ and the nominal system frequency is $f_0=\SI{50}{\hertz}$.}

We then add frequency control $u_i$ governed by our RL-DAI to each bus, where there is a connected underlying communication graph $\left(\mathcal{V},\mathcal{E}_\mathrm{Q} \right)$.
The operational cost functions are assumed to be $C_i(u_i)= c_i u_i^4/4+b_i, \forall i\in\mathcal{N},$ where the
cost coefficients $c_i$ and $b_i$ are generated randomly from $(0,1)$ and $(0,0.001)$, respectively. Thus, $\forall i\in\mathcal{N}$, $\nabla C_i(u_i)=c_i u_i^3$ by \eqref{eq:cost-power-grad} and $\zeta_i=c_i^{1/3}$ by \eqref{eq:zeta}, which determines the evolution of the integral variable $s_i$ in RL-DAI according to \eqref{eq:DAI-s-general}. Then we train the control policy $\boldsymbol{u}(\boldsymbol{s})=\left(u_i(s_i;\Xi_i(d)), i \in \mathcal{N} \right)$ as described in Section~\ref{sec:RL}.

\subsection{Training of Stacked ReLU Monotonic Neural Networks}
Following the RNN-based RL algorithm illustrated in Section~\ref{ssec:RNN-train}, we can train the stacked ReLU monotonic neural networks that construct control policy $\boldsymbol{u}(\boldsymbol{s})$ satisfying Assumption~\ref{ass-cost-grad-scale} for RL-DAI. Specifically, each sample of system trajectory is generated by disturbing random step power changes $\boldsymbol{p}$ to the system that is initially at a given supply-demand balanced setpoint, where each $p_i$ is drawn uniformly from $(-5,5)\SI{}{\pu}$. 
The values of all the hyperparameters mentioned in Section~\ref{ssec:RNN-train} are summarized as follows: $\rho=0.01$, $d=20$, $h=\SI{0.5}{\milli\second}$, $T=\SI{2.5}{\second}$, $|\mathcal{B}|=64$, $E=50$, and $\xi=0.4$ (exponentially decayed).

\subsection{Time-Domain Responses Following Power imbalances}
For the purpose of comparison, the frequency deviation
of the original system without additional frequency control $\boldsymbol{u}$ when there is a step change
of $\SI{3}{\pu}$ in power consumption at buses $13$, $21$, and $27$ is provided in Fig.~\ref{fig:sw-simu}. Clearly, in the absence of proper control, there exists noticeable steady-state frequency deviation from the nominal frequency. The performance
of the system under RL-DAI in the same scenario is given in Fig.~\ref{fig:perform-RLDAI}. Some observations are in order. First, observe from Fig.~\ref{fig:RLDAI-fre} that RL-DAI can perfectly restore the frequencies to the nominal value as predicted by Theorem~\ref{thm:equilibrium}. Second, Fig.~\ref{fig:mc_RL-DAI} confirms that RL-DAI help generators and loads asymptotically achieve identical marginal costs and thus the optimal steady-state economic dispatch problem is solved. Third, a comparison between Fig.~\ref{fig:sw-simu} and Fig.~\ref{fig:RLDAI-fre} shows that RL-DAI improves the frequency Nadir.
Last but not least, the nuances between the shapes of the trajectories of $\boldsymbol{s}$ and $\boldsymbol{u}$ in Fig.~\ref{fig:s_RL-DAI} and Fig.~\ref{fig:u_RL-DAI} are indicators of the nonlinearity of the learned control policy $\boldsymbol{u}(\boldsymbol{s})$. Clearly, this nonlinearity has not jeopardized the stability of the system.

\begin{figure}[t!]
\centering
\includegraphics[width=0.96\columnwidth]{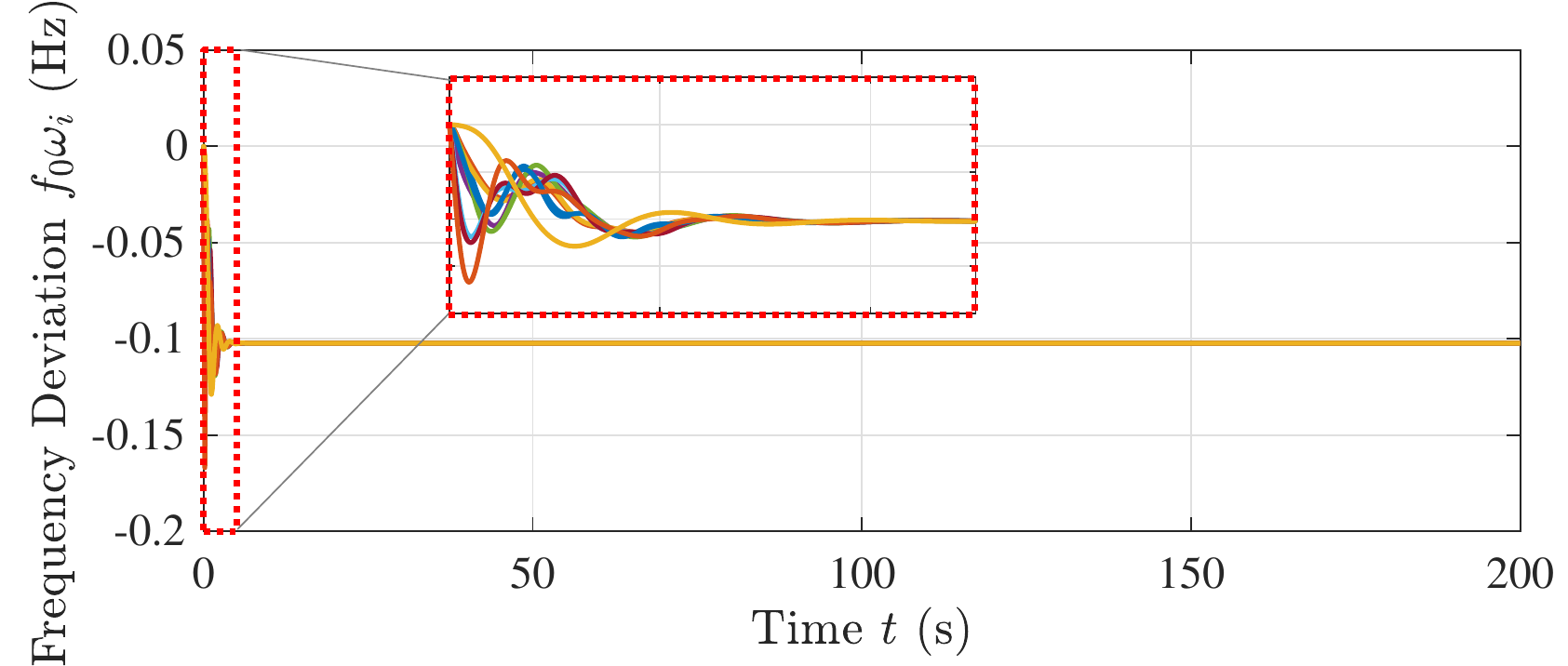}
\caption{Frequency deviations in the original system without additional frequency control $\boldsymbol{u}$ when a $\SI{-3}{\pu}$ step change in power injection is introduced to buses $13$, $21$, and $27$.}
\label{fig:sw-simu}
\end{figure}

\begin{figure}[t!]
\centering
\subfigure[Frequency deviations on generator buses]
{\includegraphics[width=0.96\columnwidth]{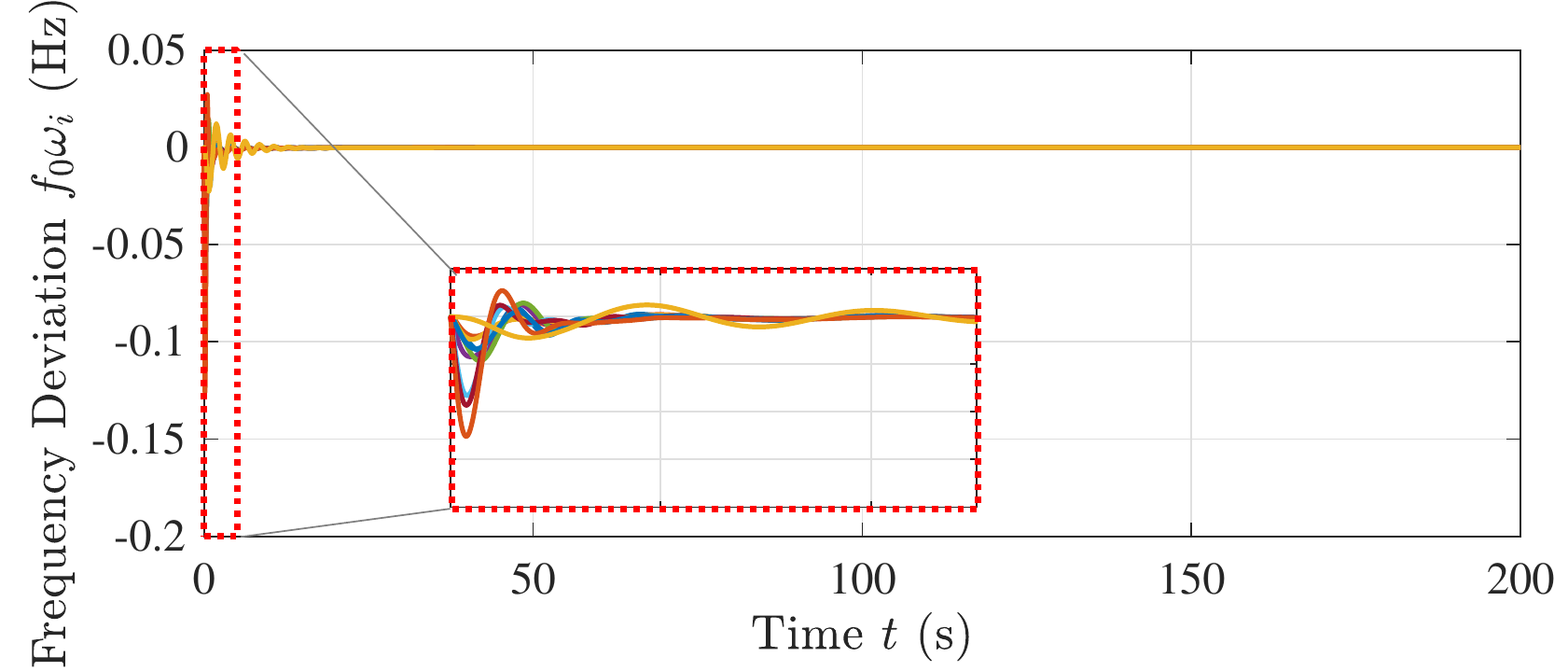}\label{fig:RLDAI-fre}}
\subfigure[Integral variables]
{\includegraphics[width=0.96\columnwidth]{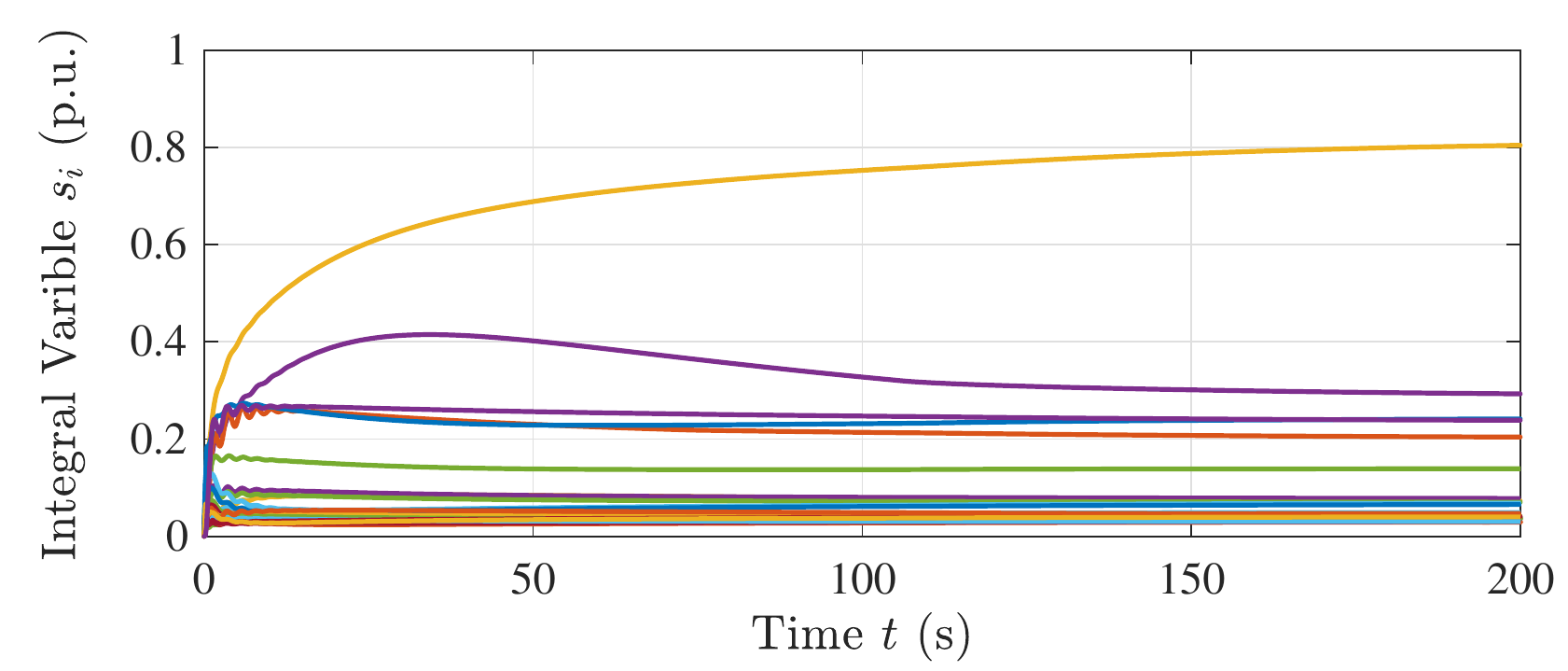}\label{fig:s_RL-DAI}}
\subfigure[Controllable power injections]
{\includegraphics[width=0.96\columnwidth]{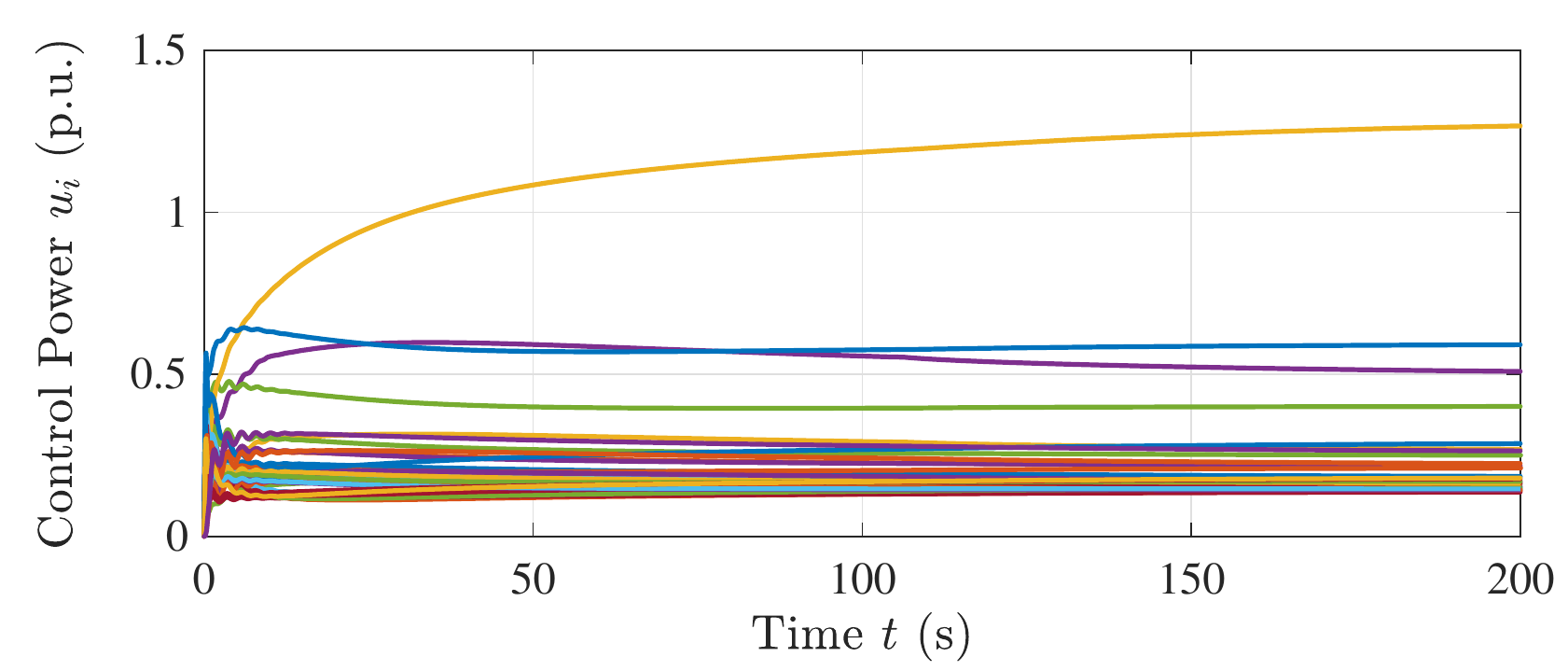}\label{fig:u_RL-DAI}}
\subfigure[Marginal costs]
{\includegraphics[width=0.96\columnwidth]{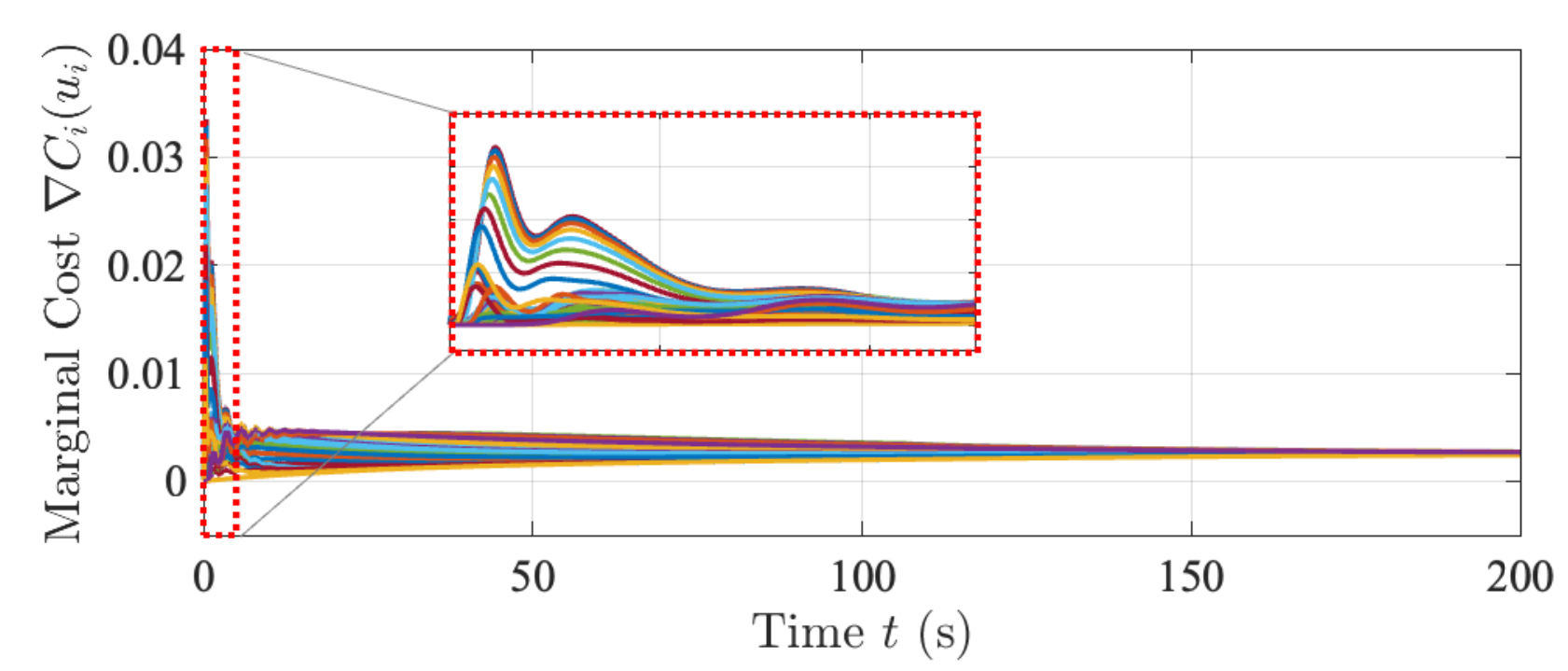}\label{fig:mc_RL-DAI}}
\subfigure[Operational costs]
{\includegraphics[width=0.96\columnwidth]{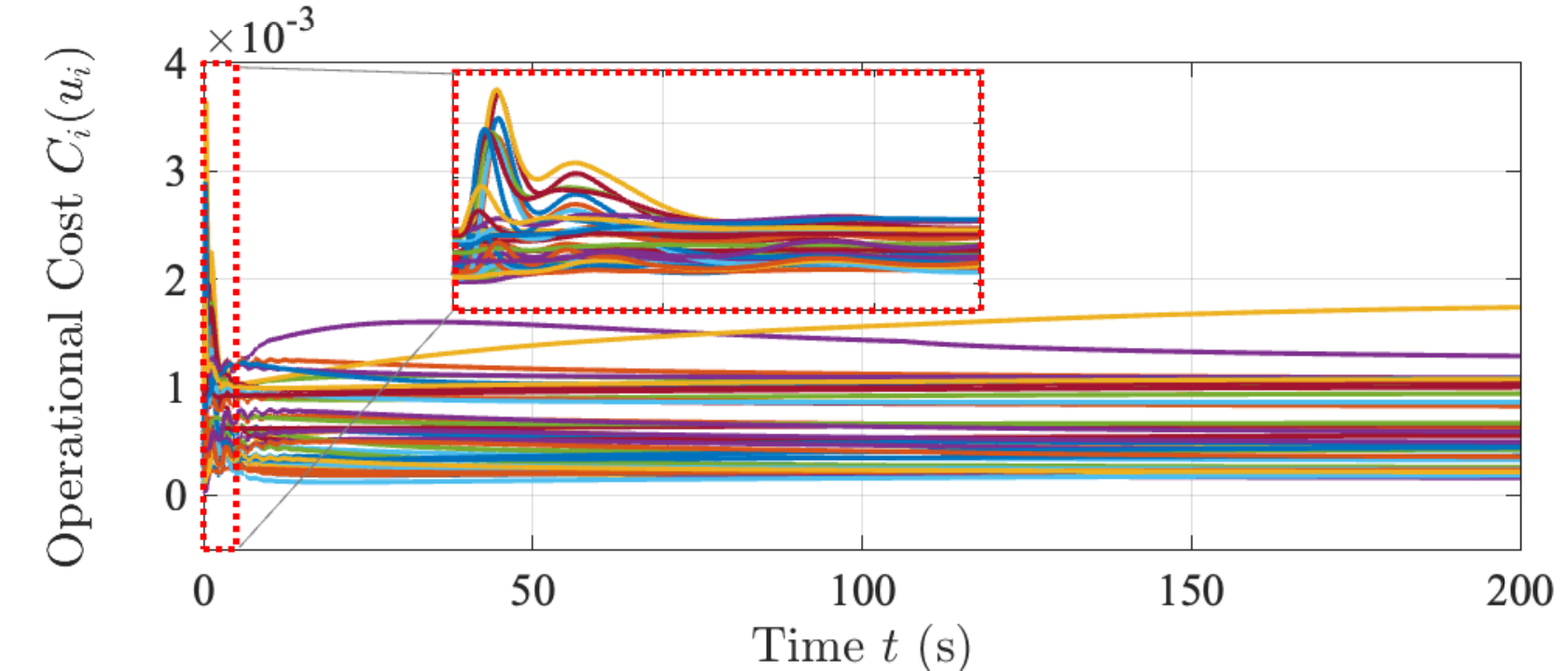}\label{fig:op-cost-RA-DAI}}
\caption{Performance of the system under RL-DAI when a $\SI{-3}{\pu}$ step change in power injection is introduced to buses $13$, $21$, and $27$.}
\label{fig:perform-RLDAI}
\end{figure}

%% file: appendix.tex
\subsection{Proof of Lemma~\ref{lem:inverse-id}}
\label{app:lem1-pf}
$\forall y\in\real$, let $\tilde{u}_i=\nabla C_i^{-1}(y)$. Clearly, $y=\nabla C_i(\tilde{u}_i)$ by the definition of the inverse function. Also, under Assumption~\ref{ass-cost-grad-scale}, $\nabla C_i(\tilde{u}_i)=\nabla C_\mathrm{o}(\zeta_i \tilde{u}_i)$, $\forall i\in\mathcal{N}$. Thus, $y=\nabla C_i(\tilde{u}_i)=\nabla C_\mathrm{o}(\zeta_i \tilde{u}_i)$, which implies that $\zeta_i \tilde{u}_i=\nabla C_\mathrm{o}^{-1}(y)$. Therefore, we have $\zeta_i \nabla C_i^{-1}(y)=\zeta_i \tilde{u}_i=\nabla C_\mathrm{o}^{-1}(y)$, $\forall i\in\mathcal{N}$,
which concludes the proof.

\subsection{Proof of Lemma~\ref{lem:bound-state}}
\label{app:lem2-pf}
By a similar argument as in the proof of~\cite[Lemma~4]{weitenberg2018exponential}, we can easily get the first inequality and 
\begin{align}\label{eq:omega-l-bound-u}
    \|\boldsymbol{\omega}_\mathcal{L}\|_2\leq\tilde{\nu}\left(\|\boldsymbol{\delta}-\boldsymbol{\delta}^\ast\|_2+\|\boldsymbol{u}(\boldsymbol{s})-\boldsymbol{u}(\boldsymbol{s}^\ast)\|_2\right)
\end{align}
for some constant $\tilde{\nu}>0$. Under Assumption~\ref{ass:nonlinear-u}, each $u_i(s_i)$ is Lipschitz continuous, which indicates that there exists some constant $\eta>0$ such that, $\forall i \in\mathcal{N}$,  
\begin{align}\label{eq:Lip-u}
    |u_i(s_i)-u_i(s_i^\ast)|\leq\eta|s_i-s_i^\ast|\,,\qquad \forall s_i\in\real\,.
\end{align}
Combining \eqref{eq:omega-l-bound-u} and \eqref{eq:Lip-u} yields
\begin{align}\label{eq:omega-l-bound}
    \|\boldsymbol{\omega}_\mathcal{L}\|_2\leq&\tilde{\nu}\left(\|\boldsymbol{\delta}-\boldsymbol{\delta}^\ast\|_2+\eta\|\boldsymbol{s}-\boldsymbol{s}^\ast\|_2\right)\nonumber\\\leq&\tilde{\nu}\left(\eta+1\right)\left(\|\boldsymbol{\delta}-\boldsymbol{\delta}^\ast\|^2_2\!+\!\|\boldsymbol{\omega}_\mathcal{G}\|^2_2\!+\!\|\boldsymbol{s}\!-\!\boldsymbol{s}^\ast\|^2_2\right)^{\frac{1}{2}}\!\,.
\end{align}
Using \eqref{eq:omega-l-bound}, we can get
\begin{align*}
    &\ \|\boldsymbol{\delta}-\boldsymbol{\delta}^\ast\|_2^2+\|\boldsymbol{\omega}\|_2^2+\|\boldsymbol{s}-\boldsymbol{s}^\ast\|_2^2\\=&\ \|\boldsymbol{\delta}-\boldsymbol{\delta}^\ast\|_2^2+\|\boldsymbol{\omega}_\mathcal{G}\|_2^2+\|\boldsymbol{s}-\boldsymbol{s}^\ast\|_2^2+ \|\boldsymbol{\omega}_\mathcal{L}\|_2^2
    \\\leq&\left[1+\tilde{\nu}^2\left(\eta+1\right)^2\right]\left(\|\boldsymbol{\delta}-\boldsymbol{\delta}^\ast\|^2_2+\|\boldsymbol{\omega}_\mathcal{G}\|^2_2+\|\boldsymbol{s}-\boldsymbol{s}^\ast\|^2_2\right)\,,
\end{align*}
which concludes the proof with $\nu:=1+\tilde{\nu}^2\left(\eta+1\right)^2$.

\subsection{Proof of Lemma~\ref{lem:strict-conv-L}}
\label{app:lem3-pf}
Let $\boldsymbol{x}:=\left(x_i, i \in \mathcal{N} \right) \in \real^n$ and $\boldsymbol{y}:=\left(y_i, i \in \mathcal{N} \right)\in\real^n$. We note that, $\forall\boldsymbol{x}\neq\boldsymbol{y}$,
\begin{align}
    \left(\nabla L (\boldsymbol{x})-\nabla L (\boldsymbol{y})\right)^T\!\!\left(\boldsymbol{x}-\boldsymbol{y}\right)=&\left(\boldsymbol{u}(\boldsymbol{x})-\boldsymbol{u}(\boldsymbol{y})\right)^T\left(\boldsymbol{x}-\boldsymbol{y}\right)\nonumber\\=&\sum_{i=1}^n\!\left(u_i(x_i)-u_i(y_i)\right)\!\left(x_i-y_i\right)\nonumber\\>&\ 0\,,\label{eq:mono-gradL}
\end{align}
where the inequality results from the monotonicity of each $u_i(s_i)$ under Assumption~\ref{ass:nonlinear-u}. Thus, by~\cite[Theorem~2.14]{rockafellar1998}, \eqref{eq:mono-gradL} indicates the strict convexity of $L(\boldsymbol{s})$.

\subsection{Proof of Lemma~\ref{lem:W}}\label{app:lem4-pf}
We can check that $W(\boldsymbol{\delta}, \boldsymbol{\omega}_\mathcal{G}, \boldsymbol{s})$ only vanishes at the equilibrium term by term. Clearly, the kinetic energy term $\pi f_0 \boldsymbol{\omega}_\mathcal{G}^T\boldsymbol{M}\boldsymbol{\omega}_\mathcal{G}>0$, $\forall \boldsymbol{\omega}_\mathcal{G}\neq \mathbbold{0}_{|\mathcal{G}|}$. By~\cite[Lemma~4]{weitenberg2018exponential}, the potential energy term $U(\boldsymbol{\delta})-U(\boldsymbol{\delta}^\ast)-\nabla U (\boldsymbol{\delta}^\ast)^T\left(\boldsymbol{\delta}-\boldsymbol{\delta}^\ast\right)\geq \beta\left\|\boldsymbol{\delta}-\boldsymbol{\delta}^{\ast}\right\|_2^2>0$, $\forall \boldsymbol{\delta}\neq \boldsymbol{\delta}^\ast$ satisfying $|\delta_i-\delta_j|\in [0,\pi/2)$, $\forall \{i,j\}\in \mathcal{E}$, for some constants $\beta>0$. By~\cite[Theorem~2.14]{rockafellar1998}, the strict convexity of $L(\boldsymbol{s})$ shown in Lemma~\ref{lem:strict-conv-L} ensures that $L(\boldsymbol{s})-L(\boldsymbol{s^\ast})-\nabla L (\boldsymbol{s^\ast})^T\left(\boldsymbol{s}-\boldsymbol{s^\ast}\right)>0$, $\forall \boldsymbol{s}\neq \boldsymbol{s}^\ast$. It follows directly that $W(\boldsymbol{\delta}^\ast, \mathbbold{0}_{|\mathcal{G}|}, \boldsymbol{s}^\ast)=0$ and $W(\boldsymbol{\delta}, \boldsymbol{\omega}_\mathcal{G}, \boldsymbol{s})>0, \forall(\boldsymbol{\delta}, \boldsymbol{\omega}_\mathcal{G}, \boldsymbol{s})\in \mathcal{D}\setminus{(\boldsymbol{\delta}^\ast, \mathbbold{0}_{|\mathcal{G}|}, \boldsymbol{s}^\ast)}$.

\subsection{Proof of Lemma~\ref{lem:Wdot}}\label{app:lem5-pf}
We calculate the derivative of $W(\boldsymbol{\delta}, \boldsymbol{\omega}_\mathcal{G}, \boldsymbol{s})$ along the trajectories of the closed-loop system \eqref{eq:sys-dyn-vec-distribute} as
\begin{align*}
&\dot{W}(\boldsymbol{\delta}, \boldsymbol{\omega}_\mathcal{G}, \boldsymbol{s})\\=&\begin{bmatrix}2\pi f_0 \boldsymbol{M}\boldsymbol{\omega}_\mathcal{G}\\\nabla U (\boldsymbol{\delta})-\nabla U (\boldsymbol{\delta}^\ast)\\\nabla L (\boldsymbol{s})-\nabla L  (\boldsymbol{s}^\ast)
\end{bmatrix}^T\begin{bmatrix}
\dot{\boldsymbol{\omega}}_\mathcal{G}\\\dot{\boldsymbol{\delta}}\\\dot{\boldsymbol{s}}
\end{bmatrix}\nonumber\\
\stackrel{\circled{1}}{=}&\ 2\pi f_0 \boldsymbol{\omega}_\mathcal{G}^T\left(-
    \boldsymbol{A}_\mathcal{G} \boldsymbol{\omega}_\mathcal{G}- \nabla_\mathcal{G} U (\boldsymbol{\delta}) + \boldsymbol{p}_\mathcal{G} + \boldsymbol{u}_\mathcal{G}(\boldsymbol{s})\right) \nonumber\\&+2\pi f_0\left(\nabla U (\boldsymbol{\delta})-\nabla U (\boldsymbol{\delta}^\ast)\right)^T\left(I_n-\frac{1}{n}\mathbbold{1}_n\mathbbold{1}_n^T\right)\boldsymbol{\omega}\nonumber\\&+ \left(\boldsymbol{u}(\boldsymbol{s})-\boldsymbol{u}(\boldsymbol{s}^\ast)\right)^T\left(-2 \pi f_0 \boldsymbol{\omega} - \boldsymbol{Z}\boldsymbol{L}_\mathrm{Q}\nabla C(\boldsymbol{u}(\boldsymbol{s}))\right)\nonumber\\
&-2\pi f_0 \boldsymbol{\omega}_\mathcal{G}^T\left(- \nabla_\mathcal{G} U (\boldsymbol{\delta}^\ast) + \boldsymbol{p}_\mathcal{G} + \boldsymbol{u}_\mathcal{G}(\boldsymbol{s}^\ast)\right)\nonumber\\
\stackrel{\circled{2}}{=}&-2\pi f_0 \boldsymbol{\omega}_\mathcal{G}^T\boldsymbol{A}_\mathcal{G} \boldsymbol{\omega}_\mathcal{G}- \left(\boldsymbol{u}(\boldsymbol{s})-\boldsymbol{u}(\boldsymbol{s}^\ast)\right)^T \boldsymbol{Z}\boldsymbol{L}_\mathrm{Q}\nabla C(\boldsymbol{u}(\boldsymbol{s}))\nonumber\\
    &-2\pi f_0 \left(\nabla_\mathcal{L} U (\boldsymbol{\delta}^\ast)-\nabla_\mathcal{L} U (\boldsymbol{\delta})+\boldsymbol{u}_\mathcal{L}(\boldsymbol{s})-\boldsymbol{u}_\mathcal{L}(\boldsymbol{s}^\ast)\right)^T\boldsymbol{\omega}_\mathcal{L}\nonumber\\
\stackrel{\circled{3}}{=}&-2\pi f_0 \boldsymbol{\omega}_\mathcal{G}^T\boldsymbol{A}_\mathcal{G} \boldsymbol{\omega}_\mathcal{G}- \boldsymbol{u}(\boldsymbol{s})^T \boldsymbol{Z}\boldsymbol{L}_\mathrm{Q}\nabla C(\boldsymbol{u}(\boldsymbol{s}))\nonumber\\
    &-2\pi f_0 \left(\nabla_\mathcal{L} U (\boldsymbol{\delta}^\ast)-\nabla_\mathcal{L} U (\boldsymbol{\delta})+\boldsymbol{u}_\mathcal{L}(\boldsymbol{s})-\boldsymbol{u}_\mathcal{L}(\boldsymbol{s}^\ast)\right)^T\boldsymbol{A}_\mathcal{L}^{-1}\nonumber\\&\quad\cdot\!\left(- \nabla_\mathcal{L} U (\boldsymbol{\delta}) + \boldsymbol{p}_\mathcal{L} + \boldsymbol{u}_\mathcal{L}(\boldsymbol{s})\right)\nonumber\\
\stackrel{\circled{4}}{=}&-2\pi f_0 \boldsymbol{\omega}_\mathcal{G}^T\boldsymbol{A}_\mathcal{G} \boldsymbol{\omega}_\mathcal{G}- \boldsymbol{u}(\boldsymbol{s})^T \boldsymbol{Z}\boldsymbol{L}_\mathrm{Q}\nabla C(\boldsymbol{u}(\boldsymbol{s}))\nonumber\\
    &-2\pi f_0 \left(\nabla_\mathcal{L} U (\boldsymbol{\delta}^\ast)-\nabla_\mathcal{L} U (\boldsymbol{\delta})+\boldsymbol{u}_\mathcal{L}(\boldsymbol{s})-\boldsymbol{u}_\mathcal{L}(\boldsymbol{s}^\ast)\right)^T\boldsymbol{A}_\mathcal{L}^{-1}\nonumber\\&\quad\cdot\!\left[- \!\nabla_\mathcal{L} U (\boldsymbol{\delta}) \!+\! \boldsymbol{p}_\mathcal{L} \!+\! \boldsymbol{u}_\mathcal{L}(\boldsymbol{s})\!-\!\left(- \!\nabla_\mathcal{L} U (\boldsymbol{\delta}^\ast) \!+\! \boldsymbol{p}_\mathcal{L} \!+\! \boldsymbol{u}_\mathcal{L}(\boldsymbol{s}^\ast)\right)\right]\nonumber\\
    =&-2\pi f_0 \boldsymbol{\omega}_\mathcal{G}^T\boldsymbol{A}_\mathcal{G} \boldsymbol{\omega}_\mathcal{G}- \boldsymbol{u}(\boldsymbol{s})^T \boldsymbol{Z}\boldsymbol{L}_\mathrm{Q}\nabla C(\boldsymbol{u}(\boldsymbol{s}))\nonumber\\
    &-2\pi f_0 \left(\nabla_\mathcal{L} U (\boldsymbol{\delta}^\ast)-\nabla_\mathcal{L} U (\boldsymbol{\delta})+\boldsymbol{u}_\mathcal{L}(\boldsymbol{s})-\boldsymbol{u}_\mathcal{L}(\boldsymbol{s}^\ast)\right)^T\boldsymbol{A}_\mathcal{L}^{-1}\nonumber\\&\quad\cdot\!\left(\nabla_\mathcal{L} U (\boldsymbol{\delta}^\ast)- \nabla_\mathcal{L} U (\boldsymbol{\delta})  + \boldsymbol{u}_\mathcal{L}(\boldsymbol{s})\!-\boldsymbol{u}_\mathcal{L}(\boldsymbol{s}^\ast)\right)\nonumber\,,
\end{align*}
which is exactly \eqref{eq:Wdot-distribute}.
Here, some tricks are used to add or remove terms for constructing a quadratic format without affecting the original value of  $\dot{W}(\boldsymbol{\delta}, \boldsymbol{\omega}_\mathcal{G}, \boldsymbol{s})$. In \circled{1}, the property that $- \nabla_\mathcal{G} U (\boldsymbol{\delta}^\ast) + \boldsymbol{p}_\mathcal{G} + \boldsymbol{u}_\mathcal{G}(\boldsymbol{s}^\ast)=\mathbbold{0}_{|\mathcal{G}|}$ by \eqref{eq:equ-DAI-delta} of Theorem~\ref{thm:equilibrium} is used. In \circled{2}, the fact that 
 $\nabla U (\boldsymbol{\delta})^T\mathbbold{1}_n=0$ is used. In \circled{3}, the fact that $\boldsymbol{\omega}_\mathcal{L}$ is determined by $\boldsymbol{\delta}$ and $\boldsymbol{s}$ through the algebraic equation \eqref{eq:sys-dyn-vec-distribute-omegaL} and the property that $\boldsymbol{u}(\boldsymbol{s}^\ast)^T \boldsymbol{Z}\boldsymbol{L}_\mathrm{Q}=\nabla C_\mathrm{o}^{-1}(\gamma)\mathbbold{1}_n^T\boldsymbol{Z}^{-1}\boldsymbol{Z}\boldsymbol{L}_\mathrm{Q}=\nabla C_\mathrm{o}^{-1}(\gamma)\mathbbold{1}_n^T\boldsymbol{L}_\mathrm{Q}= \mathbbold{0}_n^T$ obtained through \eqref{eq:equ-DAI-u} of Theorem~\ref{thm:equilibrium} are used. In \circled{4}, the property that $- \nabla_\mathcal{L} U (\boldsymbol{\delta}^\ast) + \boldsymbol{p}_\mathcal{L} + \boldsymbol{u}_\mathcal{L}(\boldsymbol{s}^\ast)=\mathbbold{0}_{|\mathcal{L}|}$ by \eqref{eq:equ-DAI-delta} of Theorem~\ref{thm:equilibrium} is used.

\subsection{Proof of Lemma~\ref{lem:bilinear-L}}\label{app:lem6-pf} 
This is a direct extension of the standard Laplacian potential function~\cite{FB-LNS}:
\begin{align*}
&\boldsymbol{x}^T \boldsymbol{Z}\boldsymbol{L}_\mathrm{Q} \boldsymbol{y}\\=& \sum_{i=1}^n \zeta_i x_i\left(\sum_{j=1}^n L_{\mathrm{Q},ij} y_j\right)\\=&\sum_{i=1}^n \zeta_i x_i\left(L_{\mathrm{Q},ii} y_i+\!\sum_{j=1,j\neq i}^n \!\!\!L_{\mathrm{Q},ij} y_j\right)\nonumber
\end{align*}
Expanding the sum, we get
\begin{align*}
&\sum_{i=1}^n \zeta_i x_i\left(\sum_{j=1,j\neq i}^n\!\!\! Q_{ij} y_i-\sum_{j=1,j\neq i}^n\!\!\! Q_{ij} y_j\right)\nonumber\\=&\sum_{i=1}^n \sum_{j=1,j\neq i}^n Q_{ij}\zeta_i x_i\left( y_i  - y_j \right)
\nonumber\\=&\sum_{i=1}^n \sum_{j=1}^n \dfrac{Q_{ij}}{2}\zeta_i x_i\left( y_i  - y_j \right)+\sum_{i=1}^n \sum_{j=1}^n \dfrac{Q_{ij}}{2}\zeta_i x_i\left( y_i  - y_j \right)\nonumber\\=&\sum_{i=1}^n \sum_{j=1}^n \dfrac{Q_{ij}}{2}\zeta_i x_i\left( y_i  - y_j \right)+\sum_{i=1}^n \sum_{j=1}^n \dfrac{Q_{ji}}{2}\zeta_j x_j\left( y_j  - y_i \right)\nonumber\\\stackrel{\circled{1}}{=}&\sum_{i=1}^n \sum_{j=1}^n \dfrac{Q_{ij}}{2}\zeta_i x_i\left( y_i  - y_j \right)+\sum_{i=1}^n \sum_{j=1}^n \dfrac{Q_{ij}}{2}\zeta_j x_j\left( y_j  - y_i \right)\nonumber\\=&\sum_{i=1}^n \sum_{j=1}^n \dfrac{Q_{ij}}{2}\left( y_i  - y_j \right)\left(\zeta_i x_i-\zeta_j x_j\right)\\=&\!\sum_{\{i,j\} \in\mathcal{E}_\mathrm{Q}}\!\!\!Q_{ij}\left( y_i  - y_j \right)\left(\zeta_i x_i-\zeta_j x_j\right)\,,
\end{align*}
where \circled{1} uses the symmetry of the Laplacian matrix $\boldsymbol{L}_\mathrm{Q}$.

\subsection{Proof of Corollary~\ref{cor:sign-cross}}\label{app:cor2-pf}
It follows directly from Lemma~\ref{lem:bilinear-L} by setting $\boldsymbol{x}=\boldsymbol{u}(\boldsymbol{s})$ and $\boldsymbol{y}=\nabla C(\boldsymbol{u}(\boldsymbol{s}))$ that
\begin{align}
    &\boldsymbol{u}(\boldsymbol{s})^T \boldsymbol{Z}\boldsymbol{L}_\mathrm{Q}\nabla C(\boldsymbol{u}(\boldsymbol{s}))\label{eq:cross-expan}\\=&\!\!\!\sum_{\{i,j\} \in\mathcal{E}_\mathrm{Q}}\!\!\!\!\!\!Q_{ij}\!\left[ \nabla C_i(u_i(s_i)) \! -\! \nabla C_j(u_j(s_j)) \right]\left[\zeta_i u_i(s_i)\!-\!\zeta_j u_j(s_j)\right]\!\,.\nonumber
\end{align}
We claim that, $\forall i,j \in\mathcal{N}$, there is parity between the signs of $\left(\nabla C_i(u_i(s_i)) \! -\! \nabla C_j(u_j(s_j))\right)$ and $\left(\zeta_i u_i(s_i)\!-\!\zeta_j u_j(s_j)\right)$, i.e.,
\begin{align}\label{eq:sign-same}
   \sign{\left( \nabla C_i(u_i)\!-\!\nabla C_j(u_j)\right)}=\sign{\left(\zeta_iu_i-\zeta_j u_j\right)}\,,
\end{align}
where $\sign$ denotes the sign function.
To see this, without loss of generality, $\forall i,j \in\mathcal{N}$, we only consider the case where $\nabla C_i(u_i)-\nabla C_j(u_j)>0$. Since $\nabla C_i(u_i)>\nabla C_j(u_j)$, we have 
\begin{align*}
\zeta_i u_i&=\zeta_i\nabla C_i^{-1}(\nabla C_i(u_i))>\zeta_i\nabla C_i^{-1}(\nabla C_j(u_j)) \\ &\stackrel{\circled{1}}{=}\nabla C_\mathrm{o}^{-1}(\nabla C_j(u_j)) \stackrel{\circled{2}}{=}\zeta_j\nabla C_j^{-1}(\nabla C_j(u_j))=\zeta_j u_j\,.
\end{align*}
Here, the inequality results from the fact that $\zeta_i>0$ and $\nabla C_i^{-1}(\cdot)$ is strictly increasing as mentioned in Remark~\ref{rem:inverse}. By Lemma~\ref{lem:inverse-id}, \circled{1} and \circled{2} hold. The cases where $\nabla C_i(u_i)-\nabla C_j(u_j)\leq0$ follow from an analogous argument. Therefore, our claim that \eqref{eq:sign-same} holds is true, which together with the fact that $Q_{ij}>0$, $\forall\{i,j\} \in\mathcal{E}_\mathrm{Q}$, implies that \eqref{eq:cross-expan} is nonnegative.
Notably, since each term of the sum in \eqref{eq:cross-expan}
is nonnegative, as the sum vanishes, each term must be $0$. Thus, $\boldsymbol{u}(\boldsymbol{s})^T \boldsymbol{Z}\boldsymbol{L}_\mathrm{Q}\nabla C(\boldsymbol{u}(\boldsymbol{s}))=0$ if and only if, $\forall\{i,j\} \in\mathcal{E}_\mathrm{Q}$,
$$\left[ \nabla C_i(u_i(s_i)) - \nabla C_j(u_j(s_j)) \right]\left[\zeta_i u_i(s_i)-\zeta_j u_j(s_j)\right]=0\,,$$
which is equivalent to $\nabla C_i(u_i(s_i))=\nabla C_j(u_j(s_j))$, $\forall\{i,j\} \in\mathcal{E}_\mathrm{Q}$ by \eqref{eq:sign-same}.
Recall that the communication graph $\left(\mathcal{V},\mathcal{E}_\mathrm{Q} \right)$ is assumed to be connected, which implies that $\nabla C_1(u_1(s_1))=\ldots=\nabla C_n(u_n(s_n))$, i.e., $\nabla C(\boldsymbol{u}(\boldsymbol{s}))\in\range{\mathbbold{1}_n}$.
 